\documentclass[11pt,oneside,a4paper]{article}

\usepackage[english]{babel}
\usepackage[a4paper,margin=2.5cm]{geometry}

\usepackage{setspace}
\usepackage{amsmath,amssymb,amsfonts,dsfont,bbm,amsthm}
\usepackage[gen,right]{eurosym}
\usepackage{textcomp,longtable,epsfig,url}
\usepackage{graphicx}
\usepackage[hidelinks,hyperfootnotes=false]{hyperref}
\usepackage[intoc]{nomencl}
\usepackage[printonlyused]{acronym}
\usepackage{verbatim}
\usepackage{tabularx}
\usepackage{url}
\usepackage{comment}
\usepackage[dvipsnames]{xcolor}
\usepackage{accents}
\usepackage{scrhack}
\usepackage{floatrow}
\usepackage{listings}
\usepackage{paralist}
\usepackage{xr}
\usepackage{tikz}
\usetikzlibrary{arrows.meta,positioning, trees}
\newcommand{\Var}{\operatorname{Var}}

\usepackage{booktabs}
\usepackage{flafter}
\usepackage{varwidth}
\usepackage{enumitem}
\usepackage{titling}
\usepackage{subcaption}
\usepackage{multirow}
\usepackage{mathtools}
\usepackage{tensor}
\usepackage{float}
\usepackage{array}

\usepackage[round,authoryear]{natbib}
\allowdisplaybreaks
\numberwithin{equation}{section}

\newcommand{\E}{\mathbb{E}}

\newcommand{\PP}{\mathbb{P}}

\newcommand{\N}{\mathcal{N}}

\newcommand{\lb}{\left(}
\newcommand{\rb}{\right)}
\newcommand{\lsb}{\left[}
\newcommand{\rsb}{\right]}

\newcommand{\bga}{\begin{align*}}
\newcommand{\eda}{\end{align*}}

\newcommand{\labeleq}[1]{\refstepcounter{equation}\tag{\theequation}\label{#1}}

\newtheorem{theorem}{Theorem}[section]
\newtheorem{lemma}[theorem]{Lemma}
\newtheorem{Proposition}[theorem]{Proposition}
\newtheorem{definition}[theorem]{Definition}

\newtheorem{example}[theorem]{Example}

\definecolor{darkgreen}{RGB}{0,120,0}

\title{Ex Ante and Scenario-Conditional Value of Information in Dynamic Portfolio Choice}

\author{Tom Pohlschmidt\thanks{Institute of Insurance Science, Ulm University, Helmholtzstr.\ 20, 89069 Ulm, Germany. Emails: tom.pohlschmidt@uni-ulm.de; shihao.zhu@uni-ulm.de. }\qquad  and  \qquad Shihao Zhu$^\ast$}

\date{\today}

\begin{document}
\maketitle

\begin{abstract}
Portfolio decisions are often made under uncertainty about expected returns, and investors revise their beliefs as new information arrives. Although richer information generally cannot reduce optimized value ex ante, {the trading strategy it induces need not improve on the strategy induced by a less informative information structure in every fixed-drift scenario.} We introduce scenario-conditional expected utility and value of information to evaluate strategies selected under the investor’s current belief. In a finite-horizon Gaussian portfolio choice model, we compare a fixed-belief investor, a price observer, and an attentive investor who also processes an auxiliary signal under an attention capacity constraint. Ex ante, price learning and auxiliary attention are valuable, and the total value of attention increases with capacity. Scenario-conditionally, learning creates a correction benefit by reducing the initial trading-position gap relative to the full-information strategy, but also a fluctuation cost through belief revisions, so the ex ante ranking can reverse. We derive the current-belief distribution of scenario-conditional value and {the associated underperformance probability, defined as the probability of negative scenario-conditional value relative to fixed-belief choice.} Over sufficiently short horizons, {the scenario-conditional value of learning is negative} when belief error is small relative to posterior uncertainty. For a given horizon, {sufficiently high attention capacity yields a higher scenario-conditional value than any fixed finite-capacity benchmark} across all drift scenarios. We also derive sharp convergence rates to full-information performance under Shannon and Tsallis technologies.

\end{abstract}




\section{Introduction} \label{intro}

Expected returns are difficult to estimate and are not directly observed when portfolio decisions are made \citep{merton_1980,kardaras2012robust}. An investor therefore begins with a belief about the drift, updates it as prices and other signals arrive, and adjusts her trading strategy dynamically. {This raises the question of how valuable the information underlying these belief revisions is.} If additional information can be ignored, a richer information structure cannot lower the optimized value function. The conventional ex ante value of information {quantifies this gain} before future signals and returns are observed \citep{berger1985decision}. In portfolio choice, it is commonly expressed in monetary or certainty-equivalent units \citep{amendinger2003monetary}. This current-belief criterion {provides a natural benchmark} for evaluating information structures, but it aggregates the induced strategy's performance across the candidate drift scenarios represented by that belief. 

The implemented strategy, however, operates under a fixed drift that is unknown to the investor, while future return shocks remain random. Ex ante dominance therefore need not imply better performance in every plausible fixed-drift scenario. This distinction is consequential when attention capacity is scarce. A portfolio manager may need to know under which drift scenarios a more informative strategy underperforms, how much probability mass the current belief assigns to those scenarios, how much capacity is required to control this downside, and how alternative information-processing technologies compare at a given capacity. Existing portfolio studies of information value focus primarily on the Bayesian objective or realized-return evaluations and provide comparatively little analytical guidance for these scenario-specific resource-allocation questions \citep{amendinger2003monetary,cvintanic2006,defranco2019bayesian}.

To address this gap, we introduce scenario-conditional expected utility for portfolio strategies. We use the resulting certainty equivalents to define scenario-conditional value of information across information structures. For each information structure, we first solve the investor's optimization problem under the current belief. We then hold the induced trading strategy and filtering rule fixed and evaluate expected utility under each candidate fixed-drift law, while continuing to average over future return shocks. The scenario-conditional expected utility evaluates one strategy, whereas the scenario-conditional value of information compares strategies induced by two information structures. Together, they provide a conditional counterpart to the prior-aggregated ex ante criterion without reoptimizing the strategy for each scenario. The closest portfolio precedent is \citet{palczewski_2015}, who compute drift-conditional utilities before averaging them; we retain the conditional comparison and derive analytical sign, distribution, horizon, and capacity results. Unlike an out-of-sample backtest \citep{demiguel2009optimal}, our evaluation keeps future shocks stochastic and varies only the fixed drift in the evaluation law.

We study the scenario-conditional value of information in a finite-horizon continuous-time portfolio choice model with an unknown constant drift, a Gaussian belief, and constant absolute risk aversion (CARA) preferences. At the current date, three investors have the same wealth and belief but use nested information structures. The naive fixed-belief investor does not update after the current date, the price observer learns from stock prices, and the attentive investor also processes an auxiliary signal whose precision is constrained by attention capacity. This nesting separates the value of price learning relative to fixed-belief choice from the incremental value of auxiliary attention relative to price learning. The entropy-based capacity constraint converts a fixed attention capacity into signal precision as posterior uncertainty evolves. The Gaussian--CARA structure yields explicit filters and optimal strategies, while deterministic posterior-variance dynamics permit exact comparisons across horizons, capacities, and information-processing technologies.

Our analysis proceeds in {four} steps. First, we solve the Bayesian portfolio problems and establish the ex ante benchmark. Price observation has strictly positive value relative to fixed-belief choice, auxiliary attention adds value relative to price observation, and the total value of attention increases strictly with capacity. Thus, the conventional information ranking holds under the investor's current belief.

Second, we use scenario-conditional value of information to show why this ex ante ranking can reverse. The scenario-conditional value of information is governed by the {initial trading-position} gap relative to the scenario-specific full-information strategy. Learning creates a correction benefit by reducing the inherited position gap, but subsequent belief revisions generate a fluctuation cost through changes in the trading position. We derive an exact ODE decomposition of these two forces for general capacity and a closed-form performance threshold for price observation. When the initial position gap is zero, there is nothing to correct and the fluctuation cost makes learning strictly worse than fixed-belief choice at every finite capacity. 
The remaining horizon further determines which force dominates. When the squared posterior-mean error is less than one half of posterior variance, price observation underperforms fixed-belief choice over sufficiently short horizons, and auxiliary attention performs worse still at any fixed positive finite capacity. By contrast, at every fixed positive horizon, sufficiently high capacity eventually dominates any given finite-capacity benchmark uniformly across drift scenarios. An exact Bayes aggregation identity links the ex ante value to the scenario-conditional values and shows how negative scenario outcomes remain consistent with positive ex ante value.

Third, the quadratic scenario-value representation yields the full current-belief distribution of scenario-conditional value as an affine transformation of a squared Gaussian variable. Its lower bound is determined by the fluctuation cost, whereas its right tail is unbounded. We obtain an analytical underperformance probability equal to the current-belief mass of scenarios in which attention performs worse than fixed-belief choice. Although capacity is exogenous in our model, this probability allows a portfolio manager to screen feasible capacity levels or technologies against a specified downside tolerance. In the numerical experiments, the underperformance probability declines sharply with capacity and is highest when the investor's current expected excess return is near zero, where the position-gap distribution is centered closest to the underperformance region.

Finally, we compare two information-processing technologies: Shannon entropy and the Tsallis family. These technologies determine how a fixed capacity is converted into signal precision as posterior uncertainty changes and therefore separate the amount of capacity from how it is deployed. We derive sharp convergence rates to the full-information benchmark as attention capacity becomes large. These rates describe the largest remaining fraction, across drift scenarios, of the performance gap associated with a fixed finite-capacity strategy. The rate is $\kappa^{-1}$ under Shannon entropy and Tsallis entropy with $1<\rho<3$, $\ln(1+\kappa)/\kappa$ at $\rho=3$, and $\kappa^{-2/(\rho-1)}$ for $\rho>3$. The transition at $\rho=3$ shows that alternative information-processing technologies resolve uncertainty at fundamentally different speeds. Consistent with these rates, the numerical experiments show that the capacities required to eliminate a common fraction of the worst-case shortfall can differ by orders of magnitude across technologies.

Taken together, scenario-conditional evaluation complements the prior-aggregated ex ante value of information with a map of portfolio performance across fixed-drift scenarios. The analysis explains why underperformance arises, identifies where and how likely it is, quantifies how horizon and capacity alter it, and compares how alternative technologies use a given attention capacity. {For portfolio management, an information source should therefore be evaluated not only by its ex ante value, but also by the scenario-conditional distribution of the performance generated by the trading strategy it induces.} The resulting diagnostics support capacity budgeting and technology comparison while complementing the conventional ex ante welfare criterion.

\subsection{Related literature}
Our paper connects three strands of literature. The first studies dynamic portfolio choice under incomplete information. \citet{gernotte_1986} establishes the separation between filtering and portfolio choice and analyzes how expected-return estimation error affects optimal allocation. \citet{karatzas2001bayesian} formulate Bayesian utility maximization when the constant drift is unobserved and investors learn from prices, while \citet{rieder2005portfolio} consider portfolio choice with an unobserved Markov-modulated drift. Subsequent work extends dynamic portfolio choice with learning and parameter uncertainty in several directions \citep{brandt_2005,baeuerle2012relaxed,baeuerle2019retirement,balter2021time,bauerle2024optimal}. \citet{lundtofte2006effect} is especially relevant because he compares a fully informed investor, an outsider who learns from returns, and an insider who additionally observes a private signal. Our information hierarchy is similar, but our question is different. We constrain the auxiliary signal through information-processing capacity and evaluate the resulting trading strategies across fixed-drift scenarios. Therefore, \citet{lundtofte2006effect} studies how exogenous signal quality affects optimal portfolio choice, whereas we characterize when learning raises or lowers the scenario-conditional value of information relative to a less informative benchmark.

A second strand studies scarce attention capacity. Following \citet{sims2003implications}, research examines how limited capacity shapes attention allocation, portfolio choice, asset prices, and market outcomes \citep{peng_2005,luo_2017,modeling_precision,miao2023asset,han2023strategic,han_2024}. The closest dynamic portfolio model is \citet{yin_2025}. They study a finite-horizon investor who learns from prices and a capacity-constrained news signal. They consider Shannon capacity and focus on portfolio allocation and the convergence of beliefs and risky-asset shares. Our analysis separates fixed-belief choice, price learning, and auxiliary attention, and evaluates the resulting trading strategies both under the investor's belief and across fixed-drift scenarios. This allows us to distinguish the ex ante value of information from its scenario-conditional value and to characterize when additional information is beneficial or detrimental across drift scenarios. For both Shannon and Tsallis entropy specifications, we also derive sharp scenario-uniform high-capacity rates for the worst-case normalized shortfall from perfect learning.

A third strand concerns the value of information and the evaluation of trading strategies. \citet{amendinger2003monetary} develop a utility-indifference monetary value for initial information. Related monetary valuations include reservation prices for initial and dynamically revealed information about a stochastic market price of risk (\citet{colanari_2021}) and the value of discrete expert opinions about a hidden Gaussian drift (\citet{gabih_2024}). 
Turning to strategy performance, \citet{cvintanic2006} derive an optimal dynamic trading strategy under parameter uncertainty and evaluate it both at estimated parameter values and along realized return histories. \citet{defranco2019bayesian} compare price-based Bayesian learning with a fixed-estimate Markowitz trading strategy using portfolio Sharpe ratios calculated within their Bayesian model.


The closest precedent for our scenario-conditional evaluation is \citet{palczewski_2015}. In a portfolio model with transaction costs, they compute utility conditional on each drift value for informed, learning, fixed-estimate, and biased investors, and then average those conditional utilities over a Gaussian drift distribution to obtain numerical information equivalents. We retain the conditional comparison before this averaging step and derive analytical sign conditions. This allows us to identify the drift scenarios in which the ex ante information ranking reverses and to relate {these} reversals to the gap between implemented and full-information trading strategies. Other approaches assess the optimality of a strategy or its out-of-sample performance. \citet{haugh2006evaluating} use a dual upper bound on the unknown optimal expected utility to assess the loss from a feasible trading strategy under a specified model, while \citet{demiguel2009optimal} estimate competing mean–variance trading strategies from historical returns and evaluate the resulting allocations on subsequent out-of-sample returns. Our scenario-conditional evaluation keeps the trading strategy and its filtering rule fixed and varies only the drift in the evaluation law.


The remainder of the paper is organized as follows. Section \ref{model} introduces the financial market, information-processing technology, and investor types. Section \ref{chap3} derives their optimal trading strategies and value functions. Section \ref{chap4} develops the ex ante benchmark and the scenario-conditional value of information. 
Section \ref{numerics} provides numerical experiments and decision implications. Section \ref{conclusion} concludes. Detailed proofs and auxiliary derivations are provided in the Appendix.

\section{The Model} \label{model}

Fix a finite horizon $T>0$ and an arbitrary current date $t\in[0,T]$. On a complete probability space $(\Omega,\mathcal{F},\mathbb{P})$, let $W^1$ and $W^2$ be independent standard Brownian motions. The financial market consists of a risk-free asset with rate $r>0$ and a stock with price process
$(S_s)_{s\in[t,T]}$ satisfying
\begin{align*}
\frac{dS_s}{S_s}=\mu ds+\sigma dW_s^1,\qquad S_t=s_t>0. \labeleq{s_base}
\end{align*}
Here $\sigma>0$ is known and $\mu\in\mathbb{R}$ is an unknown expected return.\footnote{As noted by \citet[Footnote 2]{kardaras2012robust}, volatility can be estimated relatively accurately from high-frequency data, whereas drift estimation remains difficult.} The investors considered below trade continuously in this
market.

Let $\mathcal{G}_t$ denote the information common to all investors at the
current date, including the current stock price. We condition on
$\mathcal{G}_t$ throughout and assume
\begin{align*}
\mu\mid\mathcal{G}_t\sim\N\lb m,v\rb,\qquad m\in\mathbb{R},\quad v>0. \labeleq{prior_belief}
\end{align*}
with $\mu$ conditionally independent of the future increments of $W^1$ and $W^2$. Thus, $(m,v)$ summarizes the common belief at time $t$, while only information arriving over $[t,T]$ is modeled explicitly.


Information about the unknown drift can arrive from two sources. The first is the stock-price process, and the second is an auxiliary signal to which we apply an attention capacity constraint. Whenever it is active, the auxiliary signal $(\xi_s)_{s\in[t,T]}$ satisfies
\begin{align*}
d\xi_s=\mu ds+\frac{1}{\sqrt{p^\xi_s}}dW^2_s, \labeleq{signal_with_precision}
\end{align*}
where $p^\xi_s>0$ is its instantaneous precision. We interpret $p^\xi_s=0$ as an inactive auxiliary signal. Independent measurement noise is also used for the news signal in \citet{yin_2025} and the continuous expert signal in \citet{gabih_2024}. This is a modeling benchmark for isolating the contribution of incremental information, such as analyst forecasts or independent research.

We treat price observations as a costless baseline, so the attention capacity constraint applies only to the auxiliary signal carrying information beyond stock prices.\footnote{This is a modeling normalization, not a claim that monitoring prices is literally costless. It isolates the attention capacity allocated to information beyond the traded asset's price history. Endogenizing the allocation of a common attention capacity between price observations and the auxiliary signal is left for future research.} Hence, learning from prices remains available when the auxiliary signal is inactive. We formulate the capacity constraint in terms of the posterior entropy reduction generated by the auxiliary signal. To include both Shannon and Tsallis differential entropies, we assume that $\mathcal{H}$ is translation invariant, meaning that $\mathcal{H}(\varphi(\cdot-c))=\mathcal{H}(\varphi{(\cdot)})$ for every density $\varphi$ and every $c\in\mathbb{R}$.\footnote{Shannon differential entropy is translation invariant \citep[Chapter~8]{cover_thomas_2006}. Tsallis differential entropy has the same property; see \citet{nielsen_nock_2012} for the Gaussian formula used below.} Because changing the mean of a scalar Gaussian law only translates its density, this assumption implies that its entropy depends on its variance alone. We therefore write $\mathcal{H}(\N(m,v))=\mathcal{H}(v)$.

To isolate the information supplied by the auxiliary signal, fix the current posterior variance $v_s$ and hold $p_s^\xi$ constant over $[s,s+h]$. Let $v_{s+h}$ and $v^P_{s+h}$ denote the posterior variances obtained from prices and the auxiliary signal, and from prices alone, respectively. Standard Gaussian filtering yields the instantaneous incremental entropy-reduction rate
\begin{align*}
\mathcal{I}_s^\xi
:=\lim_{h\downarrow0}\frac{\mathcal{H}(v^P_{s+h})-\mathcal{H}(v_{s+h})}{h}
=p_s^\xi v_s^2\mathcal{H}'(v_s). \labeleq{instant_information_rate}
\end{align*}
Appendix \ref{derivation_information_rate} provides the derivation of \eqref{instant_information_rate}. We impose the following conditions on $\mathcal{H}$:
\begin{align*}
v\mapsto v^2\mathcal{H}'(v)\in C^1(0,\infty),\qquad
\mathcal{H}'(v)>0,\qquad
\limsup_{u\downarrow0}\frac{1}{u\mathcal{H}'(u)}<\infty. \labeleq{entropy_regular}
\end{align*}
The positivity condition ensures that positive auxiliary precision reduces posterior entropy. The regularity and boundary conditions guarantee that the posterior-variance ODE in \eqref{filter_ai_compact} has a unique strictly positive solution over every finite horizon. 
For notational convenience, we define
\begin{align*}
f(v):=v^2\mathcal{H}'(v),\qquad v>0. \labeleq{entropy_conversion}
\end{align*}
Then $\mathcal{I}_s^\xi=p_s^\xi f(v_s)$. Following \citet{sims2003implications}, we
impose the information-processing capacity constraint
\begin{align*}
0\leq\mathcal{I}_s^\xi\leq\kappa
\quad\Longleftrightarrow\quad
0\leq p_s^\xi\leq\frac{\kappa}{f(v_s)},\qquad \kappa\geq0. \labeleq{constraint_precision}
\end{align*}
We take $\kappa$ to be exogenous and, in line with \citet{luo_2017} and \citet{han2023strategic}, impose full capacity. Hence $p_s^\xi=p^\xi(v_s)$, where
\begin{align*}
p^\xi(v):=\frac{\kappa}{f(v)}. \labeleq{full_capacity_precision}
\end{align*}

The entropy specification represents a reduced-form information-processing technology: it determines how a fixed capacity is converted into signal precision as posterior uncertainty changes. In this paper, we compare the Shannon and Tsallis specifications to isolate how this capacity--precision technology affects learning. 
\begin{example}
For Shannon entropy,
\begin{align*}
\mathcal{H}^S(v)=\frac{1}{2}\ln(2\pi ev),\qquad
f^S(v)=\frac{v}{2},\qquad
p^\xi(v)=\frac{2\kappa}{v}. \labeleq{shannon_entropy_specification}
\end{align*}
Thus, in the Shannon benchmark, auxiliary precision is inversely proportional to $v$.

For the Tsallis entropy with $\rho>1$,\footnote{Generalized models of rational inattention likewise use alternative entropy-based costs, including Tsallis entropy, to capture different information-processing frictions \citep{caplin2022rationally}.}
\begin{align*}
\mathcal{H}^{\rho}(v)&=\frac{1}{\rho-1}\lsb1-\rho^{-1/2}(2\pi v)^{(1-\rho)/2}\rsb,\\
f^{\rho}(v)&=\frac{1}{2}\rho^{-1/2}(2\pi)^{(1-\rho)/2}v^{(3-\rho)/2},\\
p^\xi(v)&=2\kappa\rho^{1/2}(2\pi)^{(\rho-1)/2}v^{(\rho-3)/2}. \labeleq{tsallis_entropy_specification}
\end{align*}
The Shannon specification is recovered as $\rho\to1$. Economically, $\rho$ determines how a fixed attention capacity translates into the precision of the auxiliary signal as posterior uncertainty changes: under full capacity, $p^\xi(v)$ increases as $v$ decreases for $1<\rho<3$, is independent of $v$ for $\rho=3$, and increases with $v$ for $\rho>3$. 
Both specifications satisfy \eqref{entropy_regular}.
\end{example}

\subsection{Investors and Belief Dynamics}

We consider three investor types: a naive fixed-belief investor ($NI$), a price observer ($PO$), and an attentive investor ($AI$). The naive investor does not update her belief after the current date, the price observer learns from stock prices, and the attentive investor additionally processes the auxiliary signal. At the current date $t$, all investors have wealth $x$ and conditional belief $\N(m,v)$. They differ only in the information processed after $t$. For $s\in[t,T]$, define
\begin{align*}
\mathcal{F}_s^{AI}&:=\mathcal{G}_t\vee
\sigma\lb S_u/S_t,\xi_u-\xi_t:t\leq u\leq s\rb,\qquad \kappa>0,\\
\mathcal{F}_s^{PO}&:=\mathcal{G}_t\vee
\sigma\lb S_u/S_t:t\leq u\leq s\rb,\\
\mathcal{F}_s^{NI}&:=\mathcal{G}_t. \labeleq{investor_filtrations}
\end{align*}
and write $\mathbb{F}^I=(\mathcal{F}_s^I)_{s\in[t,T]}$. At $\kappa=0$, we set $\mathbb{F}^{AI}=\mathbb{F}^{PO}$. For $I\in\{AI,PO,NI\}$, let
\begin{align*}
m_s^I:=\E\lsb\mu\mid\mathcal{F}_s^I\rsb,\qquad
v_s^I:=\Var\lb\mu\mid\mathcal{F}_s^I\rb. \labeleq{posterior_moments}
\end{align*}
For $AI$ and $PO$, linear-Gaussian filtering implies $\mu\mid\mathcal{F}_s^I\sim\N(m_s^I,v_s^I)$ (see, e.g., \cite{liptser2001statistics} and \cite{lundtofte2006effect}). We next introduce the belief dynamics of the three investors.

\paragraph{Attentive investor (AI).}
For $\kappa>0$, this investor observes prices and processes the auxiliary signal at full capacity as specified in \eqref{full_capacity_precision}. Her filter is
\begin{align*}
dm_s^{AI}
&=\frac{v_s^{AI}}{\sigma}d\widehat{W}_s^{1,AI}
+v_s^{AI}\sqrt{\frac{\kappa}{f(v_s^{AI})}}d\widehat{W}_s^{2,AI},
\qquad m_t^{AI}=m,\\
dv_s^{AI}&=-\lb\frac{1}{\sigma^2}+\frac{\kappa}{f(v_s^{AI})}\rb
\lb v_s^{AI}\rb^2ds,\qquad v_t^{AI}=v. \labeleq{filter_ai_compact}
\end{align*}
The variance path is deterministic and strictly positive. Define the \textit{innovation processes} by
\begin{align*}
d\widehat{W}_s^{1,AI}
=\frac{1}{\sigma}\lb\frac{dS_s}{S_s}-m_s^{AI}ds\rb,\qquad
d\widehat{W}_s^{2,AI}
=\sqrt{\frac{\kappa}{f(v_s^{AI})}}\lb d\xi_s-m_s^{AI}ds\rb
\labeleq{innovation_bm_ai}
\end{align*}
Here, $(\widehat{W}^{1,AI},\widehat{W}^{2,AI})$ is a two-dimensional $\mathbb{F}^{AI}$-Brownian motion under $\PP$. {Using \eqref{innovation_bm_ai}, we can rewrite \eqref{s_base} as}
 \begin{align*}
\frac{dS_s}{S_s}=m_s^{AI}ds+\sigma d\widehat{W}_s^{1,AI}.
\end{align*}

\paragraph{Price observer (PO).}
This investor learns from prices but does not process the auxiliary signal. Her filter is
\begin{align*}
dm_s^{PO}&=\frac{v_s^{PO}}{\sigma}d\widehat{W}_s^{1,PO},
\qquad m_t^{PO}=m,\\
dv_s^{PO}&=-\frac{\lb v_s^{PO}\rb^2}{\sigma^2}ds,
\qquad
v_s^{PO}=\frac{\sigma^2v}{\sigma^2+(s-t)v}. \labeleq{v_po_explicit}
\end{align*}
Here
\begin{align*}
d\widehat{W}_s^{1,PO}=\frac{1}{\sigma}\lb\frac{dS_s}{S_s}-m_s^{PO}ds\rb,
\end{align*}
which is an $\mathbb{F}^{PO}$-Brownian motion under $\PP$. The case $\kappa=0$
of \eqref{filter_ai_compact} coincides with this filter.

\paragraph{Naive investor (NI).}
This investor ignores all information arriving after $t$, i.e.,
$m_s^{NI}\equiv m, v_s^{NI}\equiv v, s\in[t,T].$
Since the naive investor does not update her belief after time $t$, no innovation process is introduced, and the stock-price dynamics under her belief coincide with \eqref{s_base}. 


Let $\pi=(\pi_s)_{s\in[t,T]}$ denote the dollar amount invested in the stock, with remaining wealth invested in the risk-free asset. For $I\in\{AI,PO\}$, the corresponding wealth dynamics satisfy
\begin{align*}
dX_s^{\pi,I}=\lb rX_s^{\pi,I}+\pi_s\lb m_s^I-r\rb\rb ds+\pi_s\sigma d\widehat{W}_s^{1,I}, \qquad X_t^{\pi,I}=x. \labeleq{wealth_learning_investor}
\end{align*}
For the naive investor,
\begin{align*}
dX_s^{\pi,NI}=\lb rX_s^{\pi,NI}+\pi_s\lb\mu-r\rb\rb ds+\pi_s\sigma dW_s^1, \qquad X_t^{\pi,NI}=x. \labeleq{wealth_naive_investor}
\end{align*}
where $\mu\sim\N\lb m,v\rb$.

All investors have constant absolute risk aversion (CARA) preferences,
\begin{align*}
U(x)=-\frac{1}{\gamma}\exp\lb-\gamma x\rb,\qquad \gamma>0. \labeleq{cara_utility}
\end{align*}
For convenience, define $\mathbb{R}_+:=(0,\infty)$ and $\mathcal{O}:=\mathbb{R}\times\mathbb{R}\times\mathbb{R}_+$.
The state $(t,x,m,v)$ records the current date, wealth, posterior mean, and posterior variance.

\begin{definition}\label{def_strategies}
Let $I\in\{AI,PO,NI\}$ and $(t,x,m,v)\in[0,T]\times\mathcal{O}$ be given. A trading strategy $\pi$ is admissible for investor $I$, and we write $\pi\in\mathcal{A}^I(t,x,m,v)$, if $\pi$ is $\mathbb{F}^I$-adapted and satisfies
\begin{align*}
\E^I_{t,x,m,v}\lsb\int_t^T\pi_s^2ds\rsb<\infty. \labeleq{admissible_square_integrability}
\end{align*}
For $I\in\{AI,PO\}$, we additionally require the following family of random variables to be uniformly integrable:
\begin{align*}
\left\{\exp\lb-\gamma e^{r(T-\tau)}X_\tau^{\pi,I}\rb:
\tau\text{ is an }\mathbb{F}^I\text{-stopping time in }[t,T]\right\}
\labeleq{admissible_class_D}
\end{align*}
\end{definition}
Condition \eqref{admissible_class_D} is the standard class-$(D)$ requirement in exponential-utility portfolio problems (see also \citet[Definition~3.1]{cheridito2011optimal} for a similar requirement). 
Investor $I$ chooses an admissible trading strategy to maximize her expected terminal utility:
\begin{align*}
V^I(t,x,m,v):=\sup_{\pi\in\mathcal{A}^I(t,x,m,v)}\E^I_{t,x,m,v}\lsb-\frac{1}{\gamma}\exp\lb-\gamma X_T^{\pi,I}\rb\rsb. \labeleq{portfolio_problem}
\end{align*}
All expectations in the optimization problems are taken under the probability law $\mathbb{P}$. The superscript $I$ records the investor's information structure and the corresponding wealth dynamics. 

\section{Optimal Trading Strategies}\label{chap3} 

 This section derives the optimal trading strategies and value functions for the three investors. We first solve the attentive investor's problem. The price observer then follows as the special case $\kappa=0$, while the naive investor is treated separately because she does not update her belief. Finally, we characterize the limiting case in which attention capacity becomes arbitrarily large.

\subsection{The Attentive Investor}

We now give the attentive investor's optimal trading strategy and the corresponding value function.

\begin{Proposition}\label{proposition_31}
 Let $(t,x,m,v)\in[0,T)\times\mathcal{O}$ and $\kappa\in[0,\infty)$ be fixed. Then the attentive investor's optimal trading strategy is given by
\begin{align*}
\pi_s^{\ast,AI}=e^{-r(T-s)}\frac{m_s^{AI}-r}{\gamma\sigma^2}\frac{\sigma^2}{\sigma^2+(T-s)v_s^{AI}},\qquad s\in[t,T]. \labeleq{optimal_strategy}
\end{align*}
Moreover, the corresponding value function is given by
\begin{align*}
V^{AI}(t,x,m,v;\kappa)=-\frac{1}{\gamma}\exp\lb-\gamma e^{r\lb T-t \rb}x-a(t,v)(m-r)^2-b^{AI}(t,v;\kappa)\rb, \labeleq{value_function}
\end{align*}
where
\begin{align*}
a(t,v)&=\frac{T-t}{2\lb \sigma^2+(T-t)v\rb}, \labeleq{a_explicit}\\
b^{AI}(t,v;\kappa)&=\int_t^T
\left(\frac{(v_s^{AI})^2}{\sigma^2}
+\frac{\kappa(v_s^{AI})^2}{f(v_s^{AI})}\right)
a(s,v_s^{AI})ds. \labeleq{b_integral}
\end{align*}
\end{Proposition}

\begin{proof}
The proof is given in Appendix \ref{proof_proposition_31}.
\end{proof}

The attentive investor's optimal trading strategy in \eqref{optimal_strategy} can be decomposed 
into a myopic demand and a learning-hedging component:
\begin{align*}
\pi_s^{\ast,AI}=\pi_s^{\mathrm{myop}}+\pi_s^{\mathrm{hedge}},
\end{align*}
where
\begin{align*}
\pi_s^{\mathrm{myop}}:=e^{-r(T-s)}\frac{m_s^{AI}-r}{\gamma\sigma^2},\qquad
\pi_s^{\mathrm{hedge}}:=-\frac{(T-s)v_s^{AI}}{\sigma^2+(T-s)v_s^{AI}}\pi_s^{\mathrm{myop}}. \labeleq{hedge_adjustment}
\end{align*}
The hedging component offsets part of the myopic demand and therefore reduces the magnitude of the investor's risky-asset trading strategy. This relative reduction is larger when posterior uncertainty $v_s^{AI}$ or the remaining horizon $T-s$ is larger. It disappears when the drift is known, $v_s^{AI}=0$, or at maturity, $s=T$. This pattern is standard in dynamic portfolio choice under incomplete information. See, for example, \citet{cvintanic2006}. 

The value function in \eqref{value_function} separates the current-belief term $a(t,v)(m-r)^2$ from the value $b^{AI}(t,v;\kappa)$ generated by subsequent learning. Since $a(t,v)$ decreases with $v$, greater posterior uncertainty attenuates the contribution of the current posterior risk premium $m-r$. Attention capacity affects only the subsequent variance path and hence operates through $b^{AI}$. The next proposition establishes this dependence.

\begin{Proposition}\label{propo_32}
Let $(t,x,m,v)\in[0,T)\times\mathcal{O}$ be fixed. Then
$
\kappa\mapsto b^{AI}(t,v;\kappa)
$
is strictly increasing. Consequently, $
\kappa \mapsto V^{AI}(t,x,m,v;\kappa)
$
is strictly increasing.
\end{Proposition}
\begin{proof}
See Appendix \ref{proof_propo_32}.
\end{proof}

The monotonicity of the value function in Proposition \ref{propo_32} operates through
$b^{AI}(t,v;\kappa)$.  At any fixed posterior variance, a larger $\kappa$
increases auxiliary precision and the total variance-decay rate. Along the
induced variance path, this accelerates learning and increases
$b^{AI}(t,v;\kappa)$. Hence, from the attentive investor's perspective, higher
capacity yields a higher value function.

\subsection{The Price Observer}

Setting $\kappa=0$ in the attentive-investor solution gives the price observer's optimal trading strategy and value function, so we omit a separate proof. The auxiliary signal is then inactive, while learning from stock prices remains available.

\begin{Proposition} \label{propo_33}
Let $(t,x,m,v)\in[0,T)\times\mathcal{O}$ be fixed. Then the price observer's optimal trading strategy is given by
\begin{align*}
\pi_s^{\ast,PO}=e^{-r(T-s)}\frac{m_s^{PO}-r}{\gamma\sigma^2}\frac{\sigma^2}{\sigma^2+(T-s)v_s^{PO}}, \qquad s\in[t,T]. \labeleq{strategy_price_observer}
\end{align*}
The corresponding value function is given by
\begin{align*}
V^{PO}(t,x,m,v)&=-\frac{1}{\gamma}\exp\lb-\gamma e^{r(T-t)}x-a(t,v)(m-r)^2-b^{PO}(t,v)\rb, \labeleq{value_function_price_observer}
\end{align*}
where $a(t,v)$ is given by \eqref{a_explicit}, and
\begin{align*}
b^{PO}(t,v)=\frac{1}{2}\ln\lb 1+\frac{(T-t)v}{\sigma^2}\rb-\frac{(T-t)v}{2(\sigma^2+(T-t)v)}>0. \labeleq{b_po_explicit}
\end{align*}
\end{Proposition}

The strategies in \eqref{optimal_strategy} and \eqref{strategy_price_observer} have the same feedback form. For $\kappa>0$, the variance comparison \eqref{variance_capacity_order} gives $v_s^{AI}<v_s^{PO}$ for every $s>t$. For a fixed posterior mean, the attentive investor's learning-hedge adjustment is therefore weaker and her risky position is closer in magnitude to the corresponding myopic demand. Because $m_s^{AI}$ and $m_s^{PO}$ are based on different signals, however, greater attention capacity does not imply a uniformly larger risky position along realized paths. The value functions have the same current-belief term and differ only in the contribution of subsequent learning. Since $b^{PO}(t,v)=b^{AI}(t,v;0)$, Proposition \ref{propo_32} yields, for $\kappa>0$,
\begin{align*}
b^{AI}(t,v;\kappa)>b^{PO}(t,v),\qquad V^{AI}(t,x,m,v;\kappa)>V^{PO}(t,x,m,v).
\end{align*}

\subsection{The Naive (Fixed-belief) Investor}

We next derive the naive investor's optimal trading strategy and the corresponding value function. 

\begin{Proposition}\label{propo_34}
Let $(t,x,m,v)\in[0,T)\times\mathcal{O}$ be fixed. Then the naive investor's optimal deterministic trading strategy is given by  
\begin{align*}
\pi_s^{\ast,NI}=e^{-r(T-s)}\frac{m-r}{\gamma\sigma^2}\frac{\sigma^2}{\sigma^2+(T-t)v}, \qquad s\in\lsb t,T\rsb. \labeleq{strategy_naive_investor}
\end{align*}
The corresponding value function is given by
\begin{align*}
V^{NI}(t,x,m,v)=-\frac{1}{\gamma}\exp\lb-\gamma e^{r(T-t)}x-a(t,v)(m-r)^2\rb, \labeleq{value_function_naive_investor}
\end{align*}
where $a(t,v)$ is given by \eqref{a_explicit}. 
\end{Proposition}

\begin{proof}
See Appendix \ref{propo_34_proof}.
\end{proof}

Although the naive investor does not update after time $t$, she does not treat the drift as known. She evaluates terminal utility under the Gaussian belief $\mu\sim\N(m,v)$, so drift uncertainty attenuates the trading strategy based on the posterior mean by the factor $\sigma^2/(\sigma^2+(T-t)v)$ in \eqref{strategy_naive_investor}. This fixed-belief benchmark is closest in spirit to the ambiguity-neutral pre-commitment case of \citet{balter2021time}: both retain a Gaussian belief about an unknown constant drift and optimize over deterministic no-learning policies, although \citet{balter2021time} work within a smooth-ambiguity framework. By contrast, the naive investor in \citet{palczewski_2015} sets the unknown drift equal to its initial point estimate when optimizing. Their investor therefore ignores both learning and drift uncertainty, while ours ignores learning but continues to integrate over drift uncertainty.

\subsection{The Perfect-Learning Limit}\label{chap_34}

Finally, we consider the attentive investor's perfect-learning limit as $\kappa\to\infty$. Starting from the current belief in \eqref{prior_belief}, the posterior variance converges to zero at every fixed date $s>t$ as $\kappa\to\infty$. In this precise limiting sense, posterior uncertainty is resolved immediately after time $t$, and the drift $\mu$ becomes known.

To state the limit, let $V^{FI}(t,x;\mu)$ denote the value function of an investor who knows the drift $\mu$. The corresponding value function and optimal investment strategy are given by\footnote{The derivations are given in the proof of Proposition \ref{propo_35}.}
\begin{align*}
V^{FI}(t,x;\mu)
&=-\frac{1}{\gamma}\exp\lb-\gamma e^{r(T-t)}x-a(t,0)(\mu-r)^2\rb,
\labeleq{value_function_full_info}\\
\pi_s^{\ast,FI}
&=e^{-r(T-s)}\frac{\mu-r}{\gamma\sigma^2},\qquad s\in[t,T].
\labeleq{full_information_strategy}
\end{align*}
Here, $a(t,0)=(T-t)/(2\sigma^2)$, reflecting the absence of drift uncertainty under full information.

\begin{Proposition}\label{propo_35}
Let $(t,x,m,v)\in[0,T)\times\mathcal{O}$ be fixed. The perfect-learning limit of the attentive investor satisfies
\begin{align*}
V^{AI,\infty}(t,x,m,v):=\lim_{\kappa\to\infty}V^{AI}(t,x,m,v;\kappa)=\E_{m,v}\lsb V^{FI}(t,x;\mu)\rsb, \labeleq{value_function_ai_infty}
\end{align*}
where $\E_{m,v}$ denotes expectation with respect to the current belief
$\mu\sim\N(m,v)$. Moreover,
\begin{align*}
V^{AI,\infty}(t,x,m,v)=-\frac{1}{\gamma}\exp\lb-\gamma e^{r(T-t)}x-a(t,v)(m-r)^2-b^{AI,\infty}(t,v)\rb, \labeleq{value_function_ai_infty_explicit}
\end{align*}
where $a(t,v)$ is given by \eqref{a_explicit} and
\begin{align*}
b^{AI,\infty}(t,v)&=\frac{1}{2}\ln\lb 1+\frac{(T-t)v}{\sigma^2}\rb. \labeleq{b_ai_infty}
\end{align*}
The corresponding limiting trading strategy is given, for $s\in(t,T]$, by
\begin{align*}
\pi_s^{\ast,AI,\infty}=e^{-r(T-s)}\frac{\mu-r}{\gamma\sigma^2}.
\labeleq{strategy_ai_infty}
\end{align*}
\end{Proposition}

\begin{proof}
See Appendix \ref{propo_35_proof}.
\end{proof}

Combining the preceding results, the value functions admit the following common representation relative to the naive investor:
\begin{align*}
V^{AI}(t,x,m,v;\kappa)&=V^{NI}(t,x,m,v)\exp\lb-b^{AI}(t,v;\kappa)\rb, \\
V^{PO}(t,x,m,v)&=V^{NI}(t,x,m,v)\exp\lb-b^{PO}(t,v)\rb, \labeleq{value_function_identities}\\
V^{AI,\infty}(t,x,m,v) &=V^{NI}(t,x,m,v)\exp\lb-b^{AI,\infty}(t,v)\rb.
\end{align*}
Thus, the value functions share the same current-wealth and current-belief terms and differ only in the contribution of uncertainty resolution after time $t$. This contribution is zero for the naive investor, $b^{PO}(t,v)$ for the price observer, $b^{AI}(t,v;\kappa)$ for the attentive investor, and $b^{AI,\infty}(t,v)$ in the perfect-learning limit. Moreover, for every $\kappa>0$,
\begin{align*}
0<b^{PO}(t,v)=b^{AI}(t,v;0)<b^{AI}(t,v;\kappa)<b^{AI,\infty}(t,v).
\end{align*}
Accordingly, the perfect-learning value is the limiting upper bound as attention capacity increases. These inequalities establish the ex ante Bayesian ranking directly from the value functions: price observation is valuable, auxiliary attention adds value, and value rises with capacity. Section \ref{chap4} expresses this ranking in terms of monetary certainty-equivalent values of information and then compares the ex ante evaluation with its scenario-conditional counterpart across fixed-drift scenarios.

\section{Ex Ante and Scenario-Conditional Value of Information} \label{chap4}

This section evaluates the value of information from two complementary perspectives. It first establishes the ex ante value of information under the investor's current belief. It then holds the optimal strategies from Section \ref{chap3} fixed and evaluates the corresponding scenario-conditional value of information across different fixed-drift scenarios. The scenario-conditional analysis identifies when the information ranking agrees with or reverses the ex ante ranking and isolates the mechanism behind such reversals. All values are expressed in date-$t$ certainty-equivalent units.\footnote{Certainty-equivalent differences convert expected-utility differences into date-$t$ monetary compensation. With CARA utility and deterministic discounting, each difference is the cash amount at date $t$ that makes the investor indifferent between two trading strategies. This monetary-value approach is standard in portfolio evaluations of information \citep{amendinger2003monetary,cvintanic2006,palczewski_2015}. In decision analysis, certainty-equivalent value is closely related to the buying price of information \citep{hazen1999lottery}. Under exponential utility, \citet{abbas2024value} establish the relevant equivalence across decision problems.} 

\subsection{Ex Ante Value of Information}

For $I\in\{NI,PO\}$, we define
\begin{align*}
CE^{I}_{EA}&:=e^{-r(T-t)}U^{-1}\lb V^{I}(t,x,m,v)\rb,\\
CE^{AI}_{EA}(\kappa)&:=e^{-r(T-t)}
U^{-1}\lb V^{AI}(t,x,m,v;\kappa)\rb. \labeleq{ex_ante_certainty_equivalents}
\end{align*}
The ex ante \textit{value of price observation}, \textit{total value of attention}, and \textit{incremental value of attention} are, respectively,
\begin{align*}
P^{PO}_{EA}&:=CE^{PO}_{EA}-CE^{NI}_{EA},\\
P^{AI}_{EA}(\kappa)&:=CE^{AI}_{EA}(\kappa)-CE^{NI}_{EA},\\
P^{AI}_{EA}(\kappa)-P^{PO}_{EA}
&=CE^{AI}_{EA}(\kappa)-CE^{PO}_{EA}. \labeleq{ex_ante_values}
\end{align*}

\begin{Proposition}\label{propo_41}
Let $(t,x,m,v)\in[0,T)\times\mathcal{O}$ and
$\kappa\in[0,\infty)$ be fixed. Then
\begin{align*}
P^{PO}_{EA}
&=\frac{e^{-r(T-t)}}{\gamma}b^{PO}(t,v),
\labeleq{ex_ante_observation_value_explicit}\\
P^{AI}_{EA}(\kappa)
&=\frac{e^{-r(T-t)}}{\gamma}b^{AI}(t,v;\kappa),
\labeleq{ex_ante_attention_value_explicit}\\
P^{AI}_{EA}(\kappa)-P^{PO}_{EA}
&=\frac{e^{-r(T-t)}}{\gamma}
\lb b^{AI}(t,v;\kappa)-b^{PO}(t,v)\rb.
\labeleq{ex_ante_incremental_attention_value_explicit}
\end{align*}
Moreover,
\begin{align*}
0<P^{PO}_{EA}=P^{AI}_{EA}(0)<P^{AI}_{EA}(\kappa),
\qquad \kappa>0,
\end{align*}
and $\kappa\mapsto P^{AI}_{EA}(\kappa)$ is strictly increasing.
\end{Proposition}

\begin{proof}
Equations \eqref{value_function_identities},
\eqref{ex_ante_certainty_equivalents}, and \eqref{ex_ante_values} give the
three identities.
The ordering and strict monotonicity follow from $b^{PO}>0$ and
Proposition \ref{propo_32}.
\end{proof}

Ex ante value of information is therefore unambiguous: price learning is valuable, auxiliary attention adds value, and greater capacity raises value. These quantities depend on posterior variance $v$, but not on the current posterior mean $m$, because the current-belief term $a(t,v)(m-r)^2$ is common to all three value functions. However, the scenario-conditional ordering below is fundamentally different.

\subsection{Scenario-Conditional Value of Information}

Let $\mu\in\mathbb{R}$ be an arbitrary drift and let $\PP^\mu$ denote the corresponding fixed-drift law. We write $\nu:=\N(m,v)$ for the current belief. Conditional on $\mathcal{G}_t$, the Bayesian law $\PP$ is the mixture of the fixed-drift laws under $\nu$.\footnote{More precisely, $\PP$ is the Bayesian joint law and $\{\PP^\theta\}_{\theta\in\mathbb{R}}$ is a regular conditional law given the drift. Regularity means that $\theta\mapsto\PP^\theta(A)$ is Borel measurable for every future event $A$ and, for every Borel set $B$, $\PP(A\cap\{\mu\in B\}\mid\mathcal{G}_t)=\int_B\PP^\theta(A)\nu(d\theta)$. Thus, conditional on $\mathcal{G}_t$, $\PP$ is the Bayesian mixture of its fixed-parameter components.} Because $\PP^\mu$ conditions on the drift and $\mu$ is conditionally independent of the future increments of $W^1$, the wealth dynamics under $\PP^\mu$ are given by
\begin{align*}
dX_s^\pi=\lb rX_s^\pi+\pi_s(\mu-r)\rb ds+\pi_s\sigma dW_s^1,\qquad X_t^\pi=x. \labeleq{scenario_conditional_wealth}
\end{align*}
Let $\pi^{\ast,I}$ denote investor $I$'s optimal trading strategy from Section \ref{chap3}.  The scenario-conditional expected utility is given by
\begin{align*}
J^I(t,x,m,v;\mu)&:=\E^\mu_{t,x,m,v}\lsb U\lb X_T^{\pi^{*,I}}\rb\rsb,\qquad I\in\{AI, PO, NI\},
\end{align*}
where $\E^\mu_{t,x,m,v}$ denotes expectation under the fixed-drift law $\PP^\mu$, conditional on the initial state $(t,x,m,v)$. The filtering rule and trading strategy are those selected under the current belief $\nu$ and are evaluated under $\PP^\mu$ without reoptimization. Note that the following identity holds:
\begin{align*}
V^I(t,x,m,v)
&=\int_{\mathbb{R}}J^I(t,x,m,v;\mu)\nu(d\mu),
\qquad I\in\{AI, PO, NI\},
\end{align*}
where $V^I$ is defined in \eqref{portfolio_problem}. The corresponding date-$t$ scenario-conditional certainty equivalent is \begin{align*}
CE^I_{SC}(\mu)&:=e^{-r(T-t)}U^{-1}\lb J^I(t,x,m,v;\mu)\rb,\qquad I\in\{AI, PO, NI\}. \labeleq{scenario_conditional_certainty_equivalents}
\end{align*}
For brevity, the dependence on $\kappa$ is omitted from the notation for the attentive investor. Since the full-information strategy is optimized under the fixed drift $\mu$, its scenario-conditional expected utility coincides with $V^{FI}(t,x;\mu)$. Hence, using \eqref{value_function_full_info}, the corresponding benchmark certainty equivalent is
\begin{align*}
CE_{SC}^{FI}(\mu)&:=e^{-r(T-t)}U^{-1}\left(V^{FI}(t,x;\mu)\right)=x+\frac{e^{-r(T-t)}}{\gamma}a(t,0)(\mu-r)^2. \labeleq{scenario_conditional_certainty_equivalent_fi}
\end{align*}
The scenario-conditional value of price observation, total value of attention, and incremental value of attention are
\begin{align*}
P^{PO}_{SC}(\mu)&:=CE^{PO}_{SC}(\mu)-CE^{NI}_{SC}(\mu),\\
P^{AI}_{SC}(\mu,\kappa)&:=CE^{AI}_{SC}(\mu)-CE^{NI}_{SC}(\mu), \labeleq{scenario_conditional_performance_attention}\\
P^{AI}_{SC}(\mu,\kappa)-P^{PO}_{SC}(\mu)&=CE^{AI}_{SC}(\mu)-CE^{PO}_{SC}(\mu),
\end{align*}
respectively.

To identify the mechanism behind the scenario-conditional ranking, let 
\begin{align*}
\Delta^{NI}&:=q_t^{NI}(m-r)-(\mu-r), \labeleq{trading_strategy_gaps}
\end{align*}
where $q_t^{NI}:=\sigma^2/(\sigma^2+(T-t)v)$ is the multiplicative factor in the naive investor's optimal investment strategy \eqref{strategy_naive_investor}. We refer to $\Delta^{NI}$ as the initial position gap since $\Delta^{NI}=\gamma\sigma^2e^{r(T-t)}(\pi_t^{*,NI}-\pi_t^{\ast,FI})$, so it is proportional to the difference between the naive investor’s trading strategy and the full-information benchmark.

The initial position gap also gives a direct representation of the naive investor’s scenario-conditional certainty equivalent:
\begin{align*}
    CE_{SC}^{NI}(\mu)&=x+\frac{e^{-r(T-t)}}{\gamma}\left[a(t,0)(\mu-r)^2-a(t,0)\left(\Delta^{NI}\right)^2\right]  \\
    &= CE_{SC}^{FI}(\mu)-\frac{e^{-r(T-t)}}{\gamma}a(t,0)\left(\Delta^{NI}\right)^2. \labeleq{sc_ce_ni}
\end{align*}
Hence, the certainty-equivalent shortfall relative to full information is directly proportional to the squared initial position gap. The next theorem shows how subsequent learning modifies this shortfall by reducing the effect of the inherited gap while generating an additional cost from ongoing belief updating.
  
\begin{theorem}\label{theorem_42}
Let $(t,x,m,v)\in[0,T)\times\mathcal{O}$, $\mu\in\mathbb{R}$, and $\kappa\in[0,\infty)$ be fixed. Set
\begin{align*} \lambda_s^\kappa:=q_s^{AI}v_s^{AI}\lb\frac{1}{\sigma^2}+\frac{\kappa}{f(v_s^{AI})}\rb,\qquad q_s^{AI}&:=\frac{\sigma^2}{\sigma^2+(T-s)v_s^{AI}}, \qquad s\in[t,T]. \labeleq{q_ai}
\end{align*}
Let $B(s;\kappa)$ and $C(s;\kappa)$ be the solutions of the terminal-value ODE system
\begin{align*}
\frac{d}{ds}B(s;\kappa)&=-2\lambda_s^\kappa\lb a(s,0)-B(s;\kappa)\rb+2q_s^{AI}v_s^{AI}\lambda_s^\kappa\lb a(s,0)-B(s;\kappa)\rb^2,\qquad B(T;\kappa)=0,\\
\frac{d}{ds}C(s;\kappa)&=-q_s^{AI}v_s^{AI}\lambda_s^\kappa\lb a(s,0)-B(s;\kappa)\rb,\qquad C(T;\kappa)=0. \labeleq{ode_system_correction_fluctuation}
\end{align*}
Then $B(t;\kappa)\in(0,a(t,0))$, $C(t;\kappa)\in(0,\frac{1}{2})$, and the scenario-conditional total value of attention satisfies
\begin{align*}
P^{AI}_{SC}(\mu,\kappa)=\frac{e^{-r(T-t)}}{\gamma}\lsb\underbrace{B(t;\kappa)\lb\Delta^{NI}\rb^2}_{\text{correction benefit}}-\underbrace{C(t;\kappa)}_{\text{fluctuation cost}}\rsb. \labeleq{correction_fluctuation_decomposition}
\end{align*}
Consequently,
\begin{align*}
P^{AI}_{SC}(\mu,\kappa)\gtreqqless0\Longleftrightarrow B(t;\kappa)\lb\Delta^{NI}\rb^2\gtreqqless C(t;\kappa). \labeleq{general_scenario_performance_condition}
\end{align*}
\end{theorem}
\begin{proof}
See Appendix \ref{proof_theorem_42}.
\end{proof}

Theorem \ref{theorem_42} shows that the scenario-conditional total value of attention is governed by a correction--fluctuation tradeoff. Learning creates a correction benefit by reducing the effect of the initial position gap, while subsequent belief revisions generate a fluctuation cost. Holding $t$ and $v$ fixed, the posterior mean $m$ and the fixed drift $\mu$ enter only through $\Delta^{NI}$, whereas attention capacity $\kappa$ enters only through $B(t;\kappa)$ and $C(t;\kappa)$. For the price observer, the correction benefit and fluctuation cost are available in closed form; see Appendix \ref{special_case_po}.

Combining \eqref{correction_fluctuation_decomposition} with the naive investor's certainty-equivalent shortfall in \eqref{sc_ce_ni} gives
\begin{align*}
CE_{SC}^{FI}(\mu)-CE_{SC}^{AI}(\mu)
=
\frac{e^{-r(T-t)}}{\gamma}
\lsb
\lb a(t,0)-B(t;\kappa)\rb
\lb\Delta^{NI}\rb^2
+
C(t;\kappa)
\rsb>0.
\end{align*}
Relative to the naive investor, learning therefore reduces the component of the shortfall associated with the inherited position gap from
$e^{-r(T-t)}a(t,0)\lb\Delta^{NI}\rb^2/\gamma$
to
\begin{align*}
\frac{e^{-r(T-t)}}{\gamma}
\lb a(t,0)-B(t;\kappa)\rb
\lb\Delta^{NI}\rb^2.
\end{align*}
The resulting reduction is
$\frac{e^{-r(T-t)}}{\gamma}
B(t;\kappa)\lb\Delta^{NI}\rb^2,$
which motivates the term correction benefit. At the same time, ongoing belief updating contributes the additional shortfall
$
\frac{e^{-r(T-t)}}{\gamma}C(t;\kappa),
$
which motivates the term fluctuation cost. This becomes particularly clear when $\Delta^{NI}=0$. In that case, there is no initial position gap to correct and hence no correction benefit, while the fluctuation cost remains strictly positive. This immediately yields 
\begin{align*}
P^{AI}_{SC}(\mu,\kappa)=-\frac{e^{-r(T-t)}}{\gamma}C(t;\kappa)<0.
\end{align*}
Thus, scenario-conditional value is strictly negative when the naive investor’s trading strategy already coincides with the full-information strategy.

Moreover, the quadratic structure of \eqref{correction_fluctuation_decomposition} has a direct interpretation. The fluctuation cost $C(t;\kappa)$ is independent of the fixed-drift scenario and uniformly bounded across attention capacities, whereas the correction benefit $B(t;\kappa)(\Delta^{NI})^2$ increases with the squared initial position gap. Hence, when the naive investor’s trading strategy is already close to the full-information benchmark, there is little scope for correction and the fluctuation cost can dominate. As the fixed drift moves farther from the zero-gap scenario, the potential correction benefit grows quadratically and eventually outweighs this cost.

To characterize the sign of scenario-conditional value, we introduce the underperformance region
\begin{align*}
\mathcal{D}(t,\kappa):=\left\{\mu\in\mathbb{R}:\left|\Delta^{NI}\right|<h(t,\kappa)\right\}, \labeleq{underperformance_region}
\end{align*}
where $h(t,\kappa):=\sqrt{C(t;\kappa)/B(t;\kappa)}$ is the performance threshold. By \eqref{general_scenario_performance_condition}, the underperformance region is the interval of fixed-drift scenarios satisfying $|\Delta^{NI}|<h(t,\kappa)$. It is centered at the zero-gap scenario $\mu=q_t^{NI}m+(1-q_t^{NI})r$, where the naive investor’s trading strategy coincides with the full-information benchmark, and has radius $h(t,\kappa)$. Hence, $h(t,\kappa)$ measures how far the fixed drift can deviate from the zero-gap scenario while attention still underperforms the naive strategy. Attention capacity affects the region only through its radius.

Theorem \ref{theorem_42} characterizes the scenario-conditional value of information pointwise in the fixed drift $\mu$. Combining this representation with the current belief $\nu$ yields the full distribution of scenario-conditional value across drift scenarios, which provides the key diagnostic for assessing downside exposure and performance dispersion across scenarios.

\begin{Proposition}\label{propo_43}
Let $(t,x,m,v)\in[0,T)\times\mathcal{O}$ and $\kappa\in[0,\infty)$ be fixed. Then
\begin{align*}
P_{SC}^{AI}(\mu,\kappa)\overset{d}{=}\frac{e^{-r(T-t)}}{\gamma}\left[B(t;\kappa)Z_t^2-C(t;\kappa)\right],
\labeleq{scenario_value_distribution}
\end{align*}
where
\begin{align*}
Z_t\sim\N\lb-\lb1-q_t^{NI}\rb(m-r),v\rb.
\end{align*}
In particular,
\begin{align*}
\mathbb{E}_{m,v}\left[P_{SC}^{AI}(\mu,\kappa)\right]&=\frac{e^{-r(T-t)}}{\gamma}\left[B(t;\kappa)\left(v+\lb1-q_t^{NI}\rb^2(m-r)^2\right)-C(t;\kappa)\right],\\
\operatorname{Var}_{m,v}\left(P_{SC}^{AI}(\mu,\kappa)\right)&=\frac{e^{-2r(T-t)}}{\gamma^2}B(t;\kappa)^2\left[2v^2+4v\lb1-q_t^{NI}\rb^2(m-r)^2\right]. \labeleq{mean_variance_p_sc}
\end{align*}
\end{Proposition}
\begin{proof}
See Appendix \ref{proof_propo_43}.
\end{proof}
Proposition \ref{propo_43} shows that the current-belief distribution of scenario-conditional value is an affine transformation of a squared Gaussian variable. Its lower bound is the negative fluctuation cost $-e^{-r(T-t)}C(t;\kappa)/\gamma$, while its right tail is unbounded. The coefficients $B(t;\kappa)$ and $C(t;\kappa)$ play fundamentally different roles in this distribution. By \eqref{mean_variance_p_sc}, holding $B(t;\kappa)$ fixed, the fluctuation cost $C(t;\kappa)$ lowers the expectation of scenario-conditional value without affecting its variance, thereby shifting the entire distribution to the left without changing its dispersion. In contrast, $B(t;\kappa)$ affects both the expectation and the variance by scaling the squared Gaussian component. Holding $C(t;\kappa)$ fixed, a larger $B(t;\kappa)$ raises scenario-conditional value in every nonzero-gap scenario, with larger gains in scenarios with larger $|\Delta^{NI}|$, while leaving the lower endpoint unchanged. The resulting increase in expectation reflects larger correction gains on average, while the increase in variance reflects a greater spread toward higher scenario-conditional values, since the lower endpoint remains fixed and a larger $B(t;\kappa)$ stretches the distribution away from that endpoint.

Using \eqref{underperformance_region}, the probability that the scenario-conditional total value of attention is negative is
\begin{align*}
\PP\lb P_{SC}^{AI}(\mu,\kappa)<0\rb &=\nu\lb\mathcal{D}(t,\kappa)\rb \\
&=\Phi\left(\frac{-\lb1-q_t^{NI}\rb(m-r)+h(t,\kappa)}{\sqrt{v}}\right)-\Phi\left(\frac{-\lb1-q_t^{NI}\rb(m-r)-h(t,\kappa)}{\sqrt{v}}\right),
\labeleq{negative_region_probability}
\end{align*}
where $\Phi$ denotes the standard normal distribution function. We refer to this quantity as the underperformance probability. Equation \eqref{negative_region_probability} provides a direct diagnostic for the attentive investor by quantifying the probability, under the current belief, that attention reduces scenario-conditional value relative to the naive strategy. Although attention capacity $\kappa$ is exogenous in the present model, this probability can be used to compare feasible capacity levels or information technologies according to their scenario-conditional downside exposure. 

Hence, this scenario-conditional diagnostic complements the ex ante value $P_{EA}^{AI}(\kappa)$ rather than replacing it. The ex ante value is not the current-belief expectation of $P_{SC}^{AI}(\mu,\kappa)$ because ex ante welfare aggregates scenario-conditional expected utilities before converting them into a certainty equivalent. Nevertheless, both evaluations depend on the same coefficients $B(t;\kappa)$ and $C(t;\kappa)$. To be more precise, the ex ante total value of attention can be recovered directly from the scenario-conditional values. From the definitions and using
 $U^{-1}(y)=-\ln(-\gamma y)/\gamma$, the ex ante total value of attention can be written as
\begin{align*}
P_{EA}^{AI}(\kappa)=-\frac{e^{-r(T-t)}}{\gamma}\ln\left(\frac{\int_{\mathbb{R}}J^{NI}(t,x,m,v;\mu)\exp\left(-\gamma e^{r(T-t)}P_{SC}^{AI}(\mu,\kappa)\right)\nu(d\mu)}{\int_{\mathbb{R}}J^{NI}(t,x,m,v;\mu)\nu(d\mu)}\right),
\end{align*}
which, with standard Gaussian calculations, yields
\begin{align*}
P_{EA}^{AI}(\kappa)=\frac{e^{-r(T-t)}}{\gamma}\left(\frac{1}{2}\ln\left(1+2vB(t;\kappa)\right)-C(t;\kappa)\right). \labeleq{ex_ante_attention_value_aggreagated}
\end{align*}
Comparing \eqref{ex_ante_attention_value_aggreagated} with \eqref{ex_ante_attention_value_explicit} from the previous subsection yields
\begin{align*}
b^{AI}(t,v;\kappa)=\frac{1}{2}\ln\left(1+2vB(t;\kappa)\right)-C(t;\kappa).
\end{align*}
This identity directly bridges the scenario-conditional and ex ante values of information and provides a structural interpretation of the value $b^{AI}(t,v;\kappa)$ generated by subsequent learning in Section \ref{chap3}. The scenario-conditional correction benefit aggregates into the nonlinear term $\frac{1}{2}\ln\left(1+2vB(t;\kappa)\right)$, while the fluctuation cost enters directly as $C(t;\kappa)$. Since $b^{AI}(t,v;\kappa)>0$, the aggregated correction benefit always exceeds the fluctuation cost. 

The preceding results characterize scenario-conditional value at a fixed horizon and capacity. We next examine the short-horizon behavior at fixed finite capacity.

\begin{Proposition}\label{propo_46}
Let $(x,m,v)\in\mathcal{O}$, $\mu\in\mathbb{R}$, and $\kappa>0$ be fixed, and let $\tau:=T-t$ denote the remaining horizon. The following second-order scaled limits hold:
\begin{align*}
\lim_{\tau\downarrow0}\frac{P^{PO}_{SC}(\mu;\tau)}{\tau^2}&=-\frac{v}{4\gamma\sigma^4}\lb v-2(m-\mu)^2\rb, \labeleq{short_horizon_po}\\
\lim_{\tau\downarrow0}\frac{P^{AI}_{SC}(\mu,\kappa;\tau)-P^{PO}_{SC}(\mu;\tau)}{\tau^2}&=-\frac{\kappa v}{4\gamma\sigma^2f(v)}\lb v-2(m-\mu)^2\rb. \labeleq{short_horizon_ai}
\end{align*}
Consequently, if $(m-\mu)^2<\frac{v}{2}$, then there exists $\varepsilon>0$ such that, for every $\tau\in(0,\varepsilon)$,
\begin{align*}
P^{AI}_{SC}(\mu,\kappa;\tau)<P^{PO}_{SC}(\mu;\tau)<0,
\end{align*}
where the arguments $\tau$ indicate the corresponding problem with terminal date $t+\tau$.
\end{Proposition}

\begin{proof}
See Appendix \ref{proof_propo_46}.
\end{proof}

Over short horizons, the scenario-conditional value of learning is determined by the posterior-mean error relative to posterior uncertainty. The common factor $v-2(m-\mu)^2$ captures the leading-order balance between the cost generated by belief updating and the benefit of correcting an initially misspecified belief. Auxiliary attention scales this local effect through $\kappa/f(v)$ but does not change its sign. If $(m-\mu)^2<v/2$, the initial belief is sufficiently accurate relative to its uncertainty that there is little scope for correction. The updating cost then dominates, so price observation has negative scenario-conditional value and any fixed positive finite auxiliary capacity reduces it further for sufficiently short horizons. If $(m-\mu)^2>v/2$, the initial belief error is sufficiently large relative to posterior uncertainty that the correction benefit dominates, and both price learning and auxiliary attention are locally beneficial. The boundary case $(m-\mu)^2=v/2$ is not determined by the second-order limits.

Thus, short-horizon learning is particularly valuable when the investor is confident but wrong, in the sense that the squared belief error is large relative to posterior variance. This differs from the ex ante perspective, under which greater posterior uncertainty generally increases the scope for valuable information. Scenario-conditionally, what matters locally is not uncertainty alone, but whether the fixed drift lies sufficiently far from the current posterior mean relative to that uncertainty. A relatively precise but substantially misspecified belief therefore creates strong scope for correction even over a short horizon.

\subsection{High-Capacity Convergence}\label{perfect_learning_scenario}

We now turn to the scenario-conditional implications of the perfect-learning limit established in Section \ref{chap_34}. Lemma \ref{lemma_a1} establishes that $B(t;\kappa)\to a(t,0)$ and $C(t;\kappa)\to0$ as $\kappa\to\infty$. Hence, Theorem \ref{theorem_42} immediately yields the limiting total value of attention
\begin{align*}
P^{AI,\infty}_{SC}(\mu):=\lim_{\kappa\to\infty}P^{AI}_{SC}(\mu,\kappa)=\frac{e^{-r(T-t)}}{\gamma}a(t,0)\lb\Delta^{NI}\rb^2\geq0. \labeleq{p_sc_limit}
\end{align*}
Thus, perfect learning attains the full-information benchmark in every fixed-drift scenario, and the limiting gain relative to the naive investor is determined entirely by the initial position gap. It is zero precisely when $\Delta^{NI}=0$.

The relevant question is therefore not what the limiting benchmark is, but how quickly it is approached and whether sufficiently large capacity dominates a given finite-capacity benchmark uniformly across drift scenarios. To characterize this convergence, let $r_\rho(\kappa)$ denote the rate associated with the respective attention technology, given by
\begin{align*}
r_\rho(\kappa):=
\begin{cases}
\kappa^{-1},&1\leq\rho<3,\\
\frac{\ln(1+\kappa)}{\kappa},&\rho=3,\\
\kappa^{-2/(\rho-1)},&\rho>3,
\end{cases}
\end{align*}
where $\rho=1$ corresponds to Shannon entropy and $\rho>1$ to Tsallis entropy. 

\begin{Proposition}\label{propo_high_capacity_rates}
Let $(t,x,m,v)\in[0,T)\times\mathcal{O}$, $\bar\kappa\geq0$, and $\rho\geq1$ be fixed. There exist constants $0<K_1\leq K_2<\infty$ and $\kappa_0>\bar\kappa$, independent of $\kappa$ and $\mu$, such that every $\kappa\geq\kappa_0$ satisfies
\begin{align*}
K_1r_\rho(\kappa)
\leq
\sup_{\mu\in\mathbb{R}}
\frac{
P_{SC}^{AI,\infty}(\mu)-P_{SC}^{AI}(\mu,\kappa)
}{
P_{SC}^{AI,\infty}(\mu)-P_{SC}^{AI}(\mu,\bar\kappa)
}
\leq
K_2r_\rho(\kappa).
\labeleq{voi_high_capacity_rates}
\end{align*}
Since $r_\rho(\kappa)\to0$, there exists $\widehat\kappa(\bar\kappa)>\bar\kappa$, independent of $\mu$, such that
\begin{align*}
P_{SC}^{AI}(\mu,\kappa)
>
P_{SC}^{AI}(\mu,\bar\kappa),
\qquad
\mu\in\mathbb{R},
\qquad
\kappa\geq\widehat\kappa(\bar\kappa).
\labeleq{eventual_capacity_dominance}
\end{align*}
In particular, $\bar\kappa=0$ gives uniform eventual dominance over price observation.
\end{Proposition}

\begin{proof}
See Appendix \ref{proof_propo_high_capacity_rates}.
\end{proof}

The normalized shortfall in \eqref{voi_high_capacity_rates} has a direct interpretation. For each fixed-drift scenario $\mu$, the numerator measures the remaining gap between capacity $\kappa$ and perfect learning, while the denominator measures the corresponding gap at the finite benchmark capacity $\bar\kappa$. Their ratio therefore gives the remaining shortfall relative to the benchmark shortfall. Taking the supremum over $\mu$ selects the scenario in which this relative shortfall is largest. Hence, the quantity in \eqref{voi_high_capacity_rates} measures the worst-case remaining fraction of the benchmark gap across all fixed-drift scenarios. Its convergence to zero means that the gap to perfect learning vanishes uniformly relative to the fixed-capacity benchmark. Proposition \ref{propo_high_capacity_rates} gives the sharp rate at which the worst-case normalized shortfall from perfect learning vanishes. Under Shannon entropy and Tsallis entropy with $1<\rho<3$, this shortfall decreases at rate $\kappa^{-1}$. At $\rho=3$, the rate is $\ln(1+\kappa)/\kappa$, while $\rho>3$ yields the slower polynomial rate $\kappa^{-2/(\rho-1)}$.

The result also strengthens the limiting comparison in this subsection. At every fixed finite capacity $\bar{\kappa}$, sufficiently large attention capacity delivers a strictly larger scenario-conditional value of information in every drift scenario. In particular, sufficiently large capacity uniformly dominates price-only learning. This eventual dominance does not imply that the scenario-conditional value of attention is globally increasing in $\kappa$. At intermediate capacities, the joint effect of the correction benefit and fluctuation cost need not be monotone in capacity.

In a related finite-horizon portfolio model, \citet{yin_2025} study how Shannon channel capacity affects the investment strategy and its convergence. Proposition \ref{propo_high_capacity_rates} instead establishes scenario-uniform convergence rates for the normalized shortfall from perfect learning, covers both Shannon and Tsallis information technologies, and identifies the rate transition at $\rho=3$.

\section{Numerical Experiments and Decision Implications} \label{numerics}

An investor who can process an auxiliary signal faces a practical question: when does additional information improve scenario-conditional performance, and when does it instead amplify costly trading fluctuations? We address this question through four related decisions: how capacity changes the correction benefit and fluctuation cost for different initial position gaps, how much capacity is required to control scenario-conditional underperformance, how the investment horizon changes the relevant performance comparison, and how the information-processing technology affects convergence to perfect learning.

\begin{table}[htbp]
    \centering
      \caption{Baseline parameters used in the numerical analysis.}
          \begin{tabular}{ccc}
        \toprule
        Parameter & Value & Interpretation \\
        \midrule
        $t$ & $0$ & Initial date \\
        $T$ & $5$ & Maturity \\
        $m$ & $0.08$ & Initial posterior mean \\
        $v$ & $0.09$ & Initial posterior variance \\
        $\kappa$ & $0.5$ & Active attention capacity \\
        $\sigma$ & $0.3$ & Stock-return volatility \\
        $r$ & $0.04$ & Risk-free rate \\
        $\mu$ & $0.09$ & {Fixed drift scenario} \\
        $\gamma$ & $0.02$ & Absolute risk aversion \\
        $\rho$ & $1$ & Entropy index (Shannon) \\
        \bottomrule
    \end{tabular}
    \label{basecase_parameters}
\end{table}

These results are comparative-static illustrations, and all values are reported in certainty-equivalent units. All numerics are performed in R, and the code is available upon request. Unless stated otherwise, parameters are fixed at the baseline values in Table \ref{basecase_parameters}. The five-year horizon, expected risk premium of four percentage points, and stock-return volatility $\sigma=0.3$ provide an illustrative investment setting. The prior variance $v=0.09$ reflects substantial drift uncertainty, with $\sigma^2/v=1$ year providing a natural price-learning time scale. At the Shannon benchmark $\kappa=0.5$, the initial auxiliary precision $2\kappa/v$ equals the price-observation precision $1/\sigma^2$. The fixed drift $\mu=0.09$ lies only one percentage point above the posterior mean, but the uncertainty adjustment $q_t^{NI}=1/6$ yields an initial position gap $\Delta^{NI}=-0.0433$. Thus, the naive investor initially holds a substantially smaller risky position than the full-information investor. The risk-aversion parameter $\gamma=0.02$ determines the monetary scale of the reported values without affecting their signs.

\subsection{Correction Benefit and Fluctuation Cost}

The first question is: how does additional attention capacity affect performance when the current trading position is at different distances from its scenario-specific benchmark? Figure \ref{tradeoffs_p_sc} answers this question by varying capacity at three fixed initial position gaps. Panel (a) separates the correction benefit from the fluctuation cost, and Panel (b) reports their net effect on the scenario-conditional total value of attention.

\begin{figure}[t]
\centering
\begin{subfigure}{0.48\linewidth}
\centering
\includegraphics[width=\linewidth]{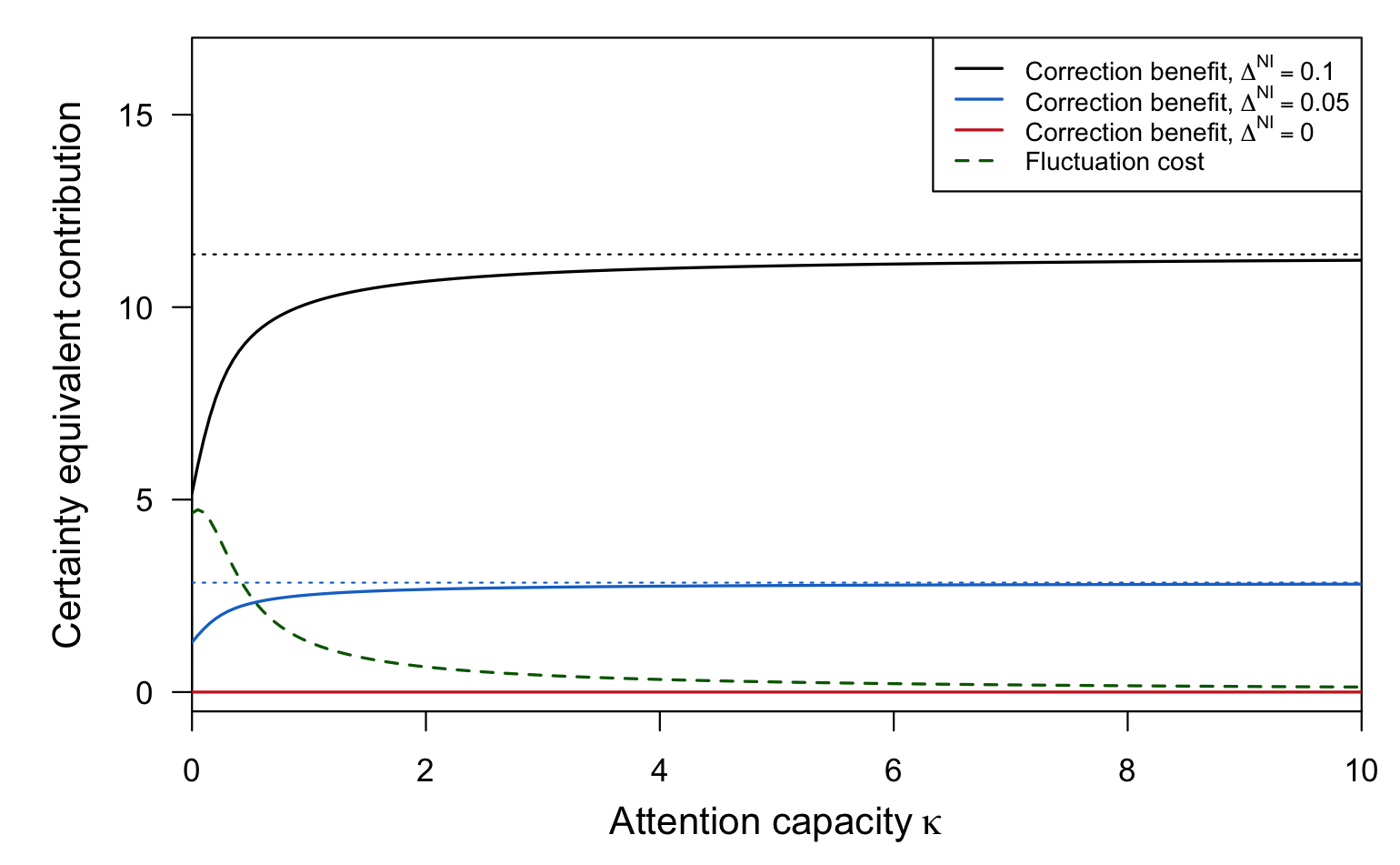}
\caption{Correction benefit and fluctuation cost.}
\label{figure_1a}
\end{subfigure}
\hfill
\begin{subfigure}{0.48\linewidth}
\centering
\includegraphics[width=\linewidth]{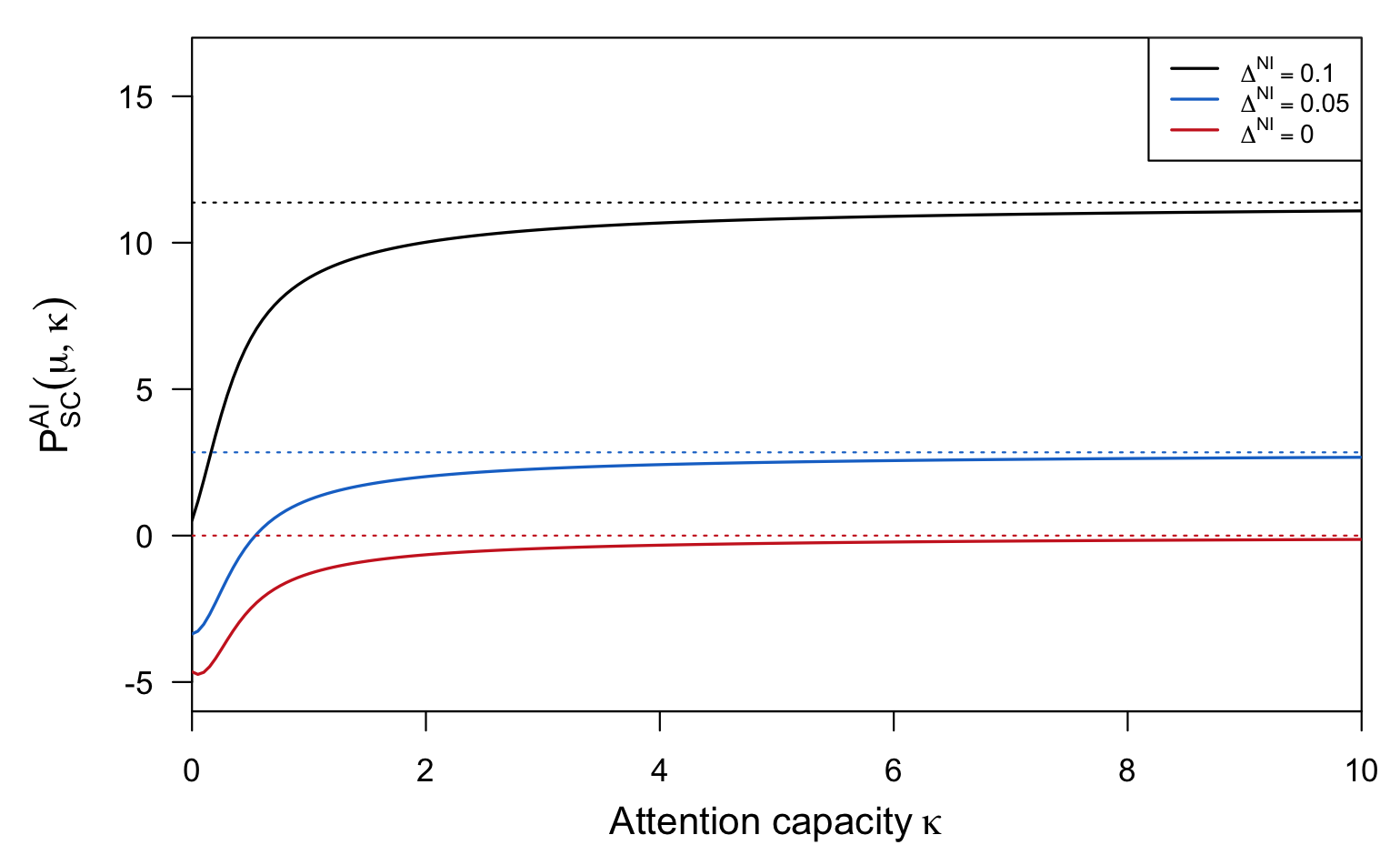}
\caption{Scenario-conditional total value of attention.}
\label{figure_1b}
\end{subfigure}
\caption{Correction benefit, fluctuation cost, and scenario-conditional total value of attention over capacity $\kappa$ for fixed initial position gaps $\Delta^{NI}\in\{0,0.05,0.1\}$. Panel (a) reports the two certainty-equivalent components in Theorem \ref{theorem_42}. Panel (b) reports their net value $P^{AI}_{SC}(\mu,\kappa)$. Capacity $\kappa=0$ corresponds to price observation, dotted horizontal lines mark the perfect-learning limits, and all remaining parameters are fixed at their baseline values.
}
\label{tradeoffs_p_sc}
\end{figure}

From Panel (a), we find that the correction benefit depends strongly on the initial position gap. It is zero when $\Delta^{NI}=0$, and the curve for $\Delta^{NI}=0.1$ is four times the curve for $\Delta^{NI}=0.05$ because the benefit is proportional to $(\Delta^{NI})^2$. Panel (a) also reveals a low-capacity nonmonotonicity in the fluctuation cost: introducing a weak auxiliary signal first raises the cost slightly, whereas further capacity lowers it toward zero.

From Panel (b), we find that these two effects produce different performance patterns across initial positions. When $\Delta^{NI}=0$, the total value remains negative at every displayed finite capacity. When $\Delta^{NI}=0.05$, the value is initially negative and becomes positive within the displayed capacity range. When $\Delta^{NI}=0.1$, the value is positive at every displayed positive capacity. Thus, the same capacity can reduce performance for a position that is already aligned with its scenario-specific benchmark and improve performance for a position with a sufficiently large initial gap.

The capacity effects reflect the speed at which posterior uncertainty is resolved. Greater capacity allows the attentive investor to correct an inherited position gap more rapidly, which raises the correction benefit when $\Delta^{NI}\neq0$. At the same time, faster resolution of uncertainty shortens the period over which posterior revisions generate fluctuations in the investment strategy. In this calibration, this reduction in the duration of belief-driven fluctuations dominates their greater initial intensity, causing the fluctuation cost to decline with capacity. 


The decision implication is that attention capacity should be evaluated jointly with the position it is expected to correct. A small increase in capacity need not improve scenario-conditional performance when the inherited position is already close to its benchmark. 

\subsection{Distribution of the Scenario-Conditional Total Value of Attention}

The second question is: how much capacity is required to keep the probability of underperforming the fixed-belief benchmark below a specified target? Figure \ref{distribution_p_sc} first shows how this probability changes with capacity and the posterior risk premium. Table \ref{underperformance_table} then translates selected probability targets into capacity requirements across posterior variances.

\begin{figure}[thbp]
\centering
\begin{subfigure}{0.48\linewidth}
\centering
\includegraphics[width=\linewidth]{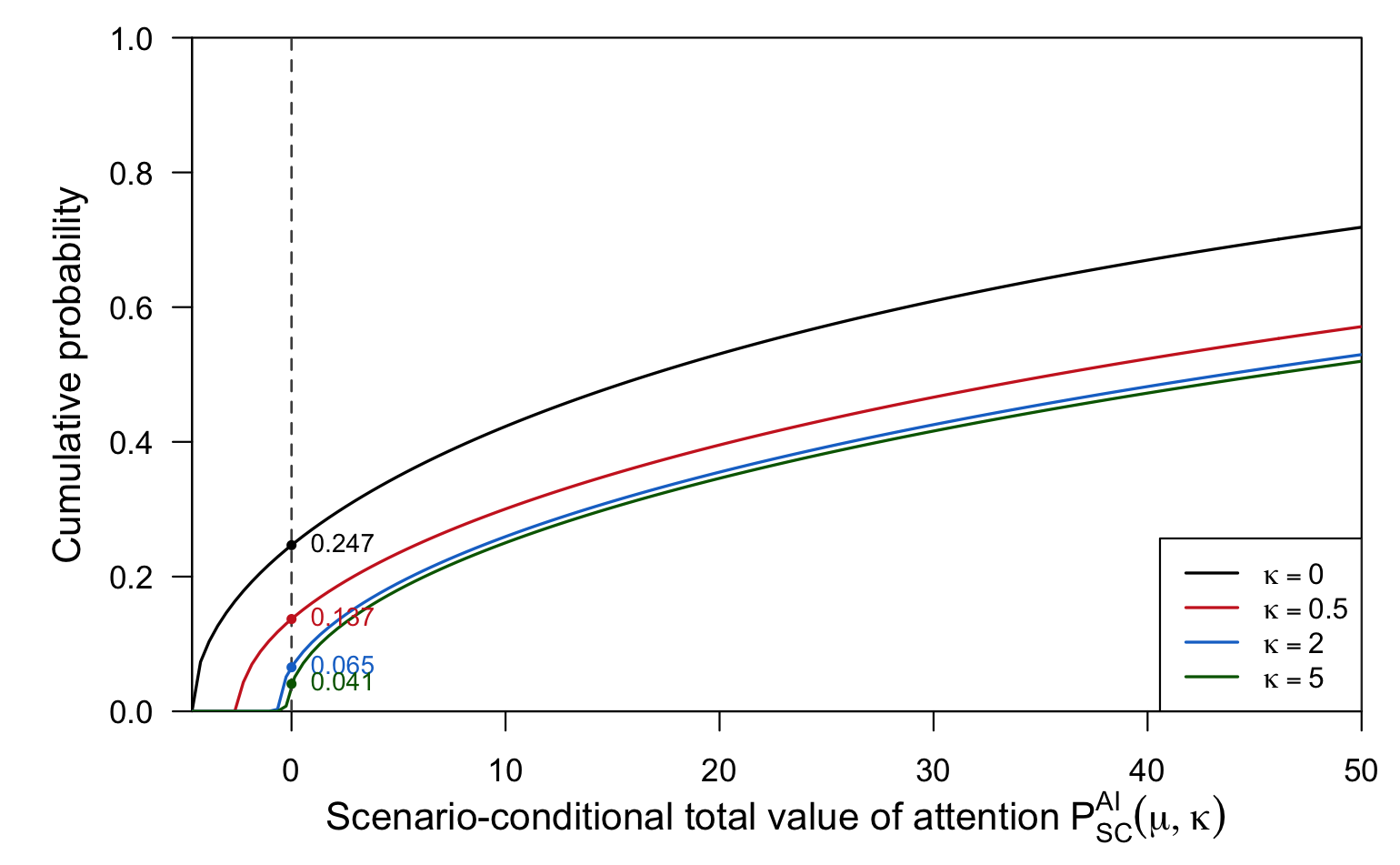}
\caption{CDF of scenario-conditional total value.}
\label{figure_2a}
\end{subfigure}
\hfill
\begin{subfigure}{0.48\linewidth}
\centering
\includegraphics[width=\linewidth]{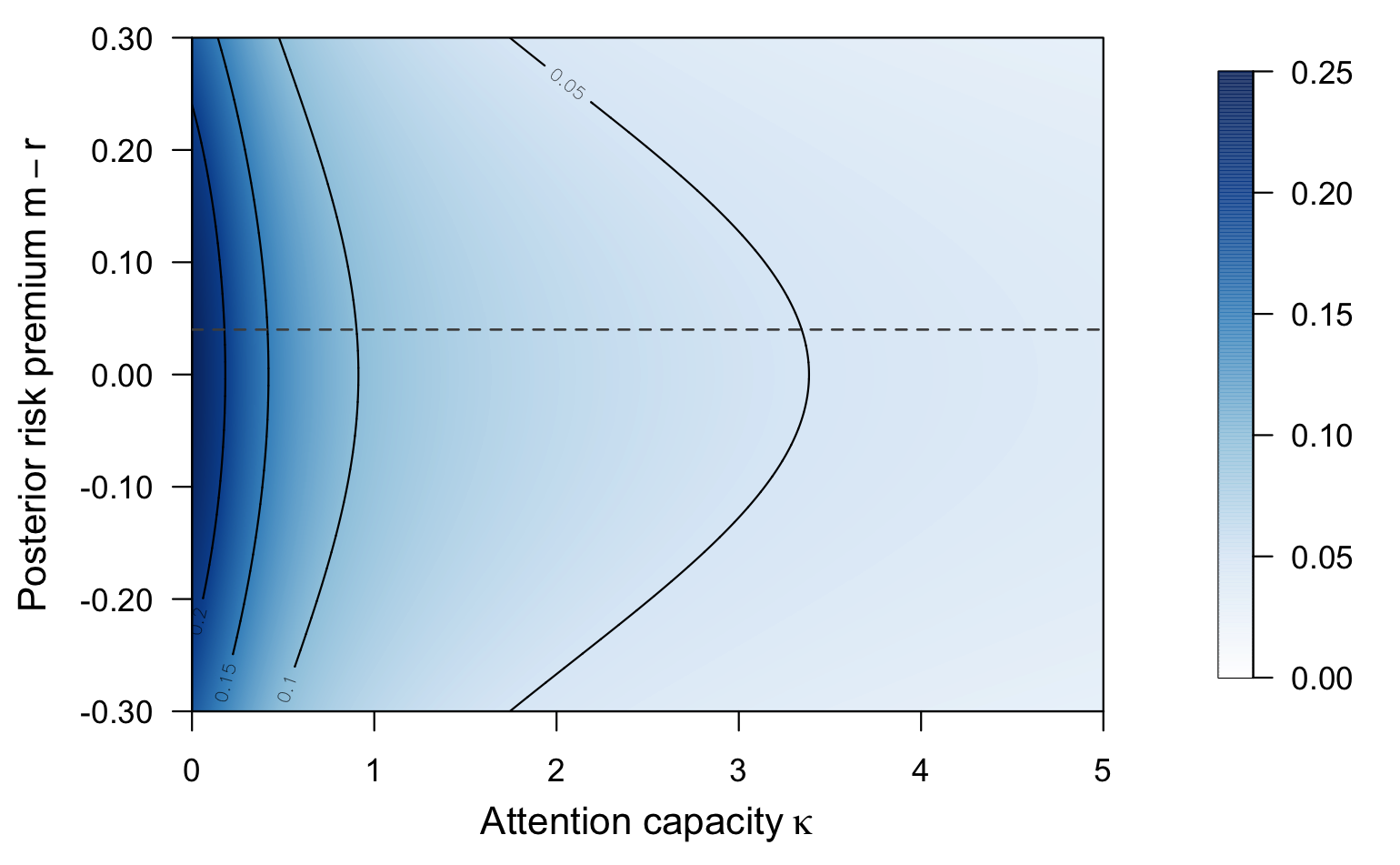}
\caption{Underperformance probability.}
\label{figure_2b}
\end{subfigure}
\caption{CDF of the scenario-conditional total value of attention and the corresponding underperformance probability. Panel (a) reports the CDF of $P^{AI}_{SC}(\mu,\kappa)$ for $\kappa\in\{0,0.5,2,5\}$; each CDF evaluated at zero is the probability of underperforming the fixed-belief benchmark. Panel (b) reports this probability from \eqref{negative_region_probability} over capacity $\kappa$ and posterior risk premium $m-r$. Contours mark selected probability levels, the dashed line marks the baseline posterior risk premium, and all remaining parameters are fixed at their baseline values.
}
\label{distribution_p_sc}
\end{figure}

\begin{table}[t]
\centering
\begin{tabular}{lrrrr}
\toprule
& \multicolumn{4}{c}{Posterior variance $v$} \\
\cmidrule(lr){2-5}
Underperformance probability
& $0.01$ & $0.04$ & $0.09$ & $0.16$ \\
\midrule
$20\%$ & 0.621   & 0.374   & 0.178  & 0.000 \\
$10\%$ & 2.350   & 1.412   & 0.903  & 0.616 \\
$5\%$  & 9.246   & 5.411   & 3.345  & 2.243 \\
$1\%$  & 229.891 & 133.175 & 80.821 & 52.889 \\
\bottomrule
\end{tabular}
\caption{Minimum attention capacity required to reduce the underperformance probability to at most the indicated target level across posterior variances. All remaining parameters are fixed at their baseline values. A value of zero indicates that price observation alone already satisfies the target.}
\label{underperformance_table}
\end{table}


From Panel (a), we find that the probability of a negative scenario-conditional total value, measured by each CDF at zero, falls from $0.247$ under price observation to $0.041$ at $\kappa=5$. Panel (a) also shows that the change is not a simple location shift. In the current Gaussian model, the distribution has a capacity-dependent finite lower endpoint and an unbounded right tail. Across the displayed capacities, the lower endpoint moves toward zero while the distribution stretches toward larger positive values, consistent with the distinct roles of the fluctuation cost and correction coefficient in Proposition \ref{propo_43}.



Panel (b) shows that the underperformance probability depends strongly on both attention capacity and the posterior risk premium. At a given capacity, the probability is highest around $m-r=0$ and decreases symmetrically as $|m-r|$ increases. This symmetry follows directly from \eqref{negative_region_probability}. When $m-r=0$, the distribution of the initial position gap is centered at zero, where the underperformance region in \eqref{underperformance_region} is located. As $|m-r|$ increases, the position-gap distribution moves away from zero, increasing the probability of scenarios in which the correction benefit dominates the fluctuation cost.

Increasing attention capacity sharply contracts the high-probability region in Panel (b). Over the displayed range, Figure \ref{tradeoffs_p_sc} shows a declining fluctuation cost and an increasing correction coefficient, so the performance threshold $h(t,\kappa)$ falls and the underperformance region narrows. Consequently, progressively larger portions of the posterior-risk-premium range are associated with low probabilities of negative scenario-conditional value.

From Table \ref{underperformance_table}, we find that tighter downside targets require capacity to rise nonlinearly. At the baseline posterior variance, a $10\%$ target requires $\kappa=0.903$, whereas a $1\%$ target requires $\kappa=80.821$. Reading across each row also shows that the required capacity decreases as posterior variance rises in this calibration because the position gap becomes more dispersed relative to the underperformance threshold.

The decision implication is that capacity requirements should be reported for an explicit underperformance target and an explicit posterior variance. The lower relative underperformance probability at a larger posterior variance does not imply lower absolute portfolio risk, and it does not imply that an investor prefers greater uncertainty.



\begin{figure}[thbp]
\centering
\begin{subfigure}{0.48\linewidth}
\centering
\includegraphics[width=\linewidth]{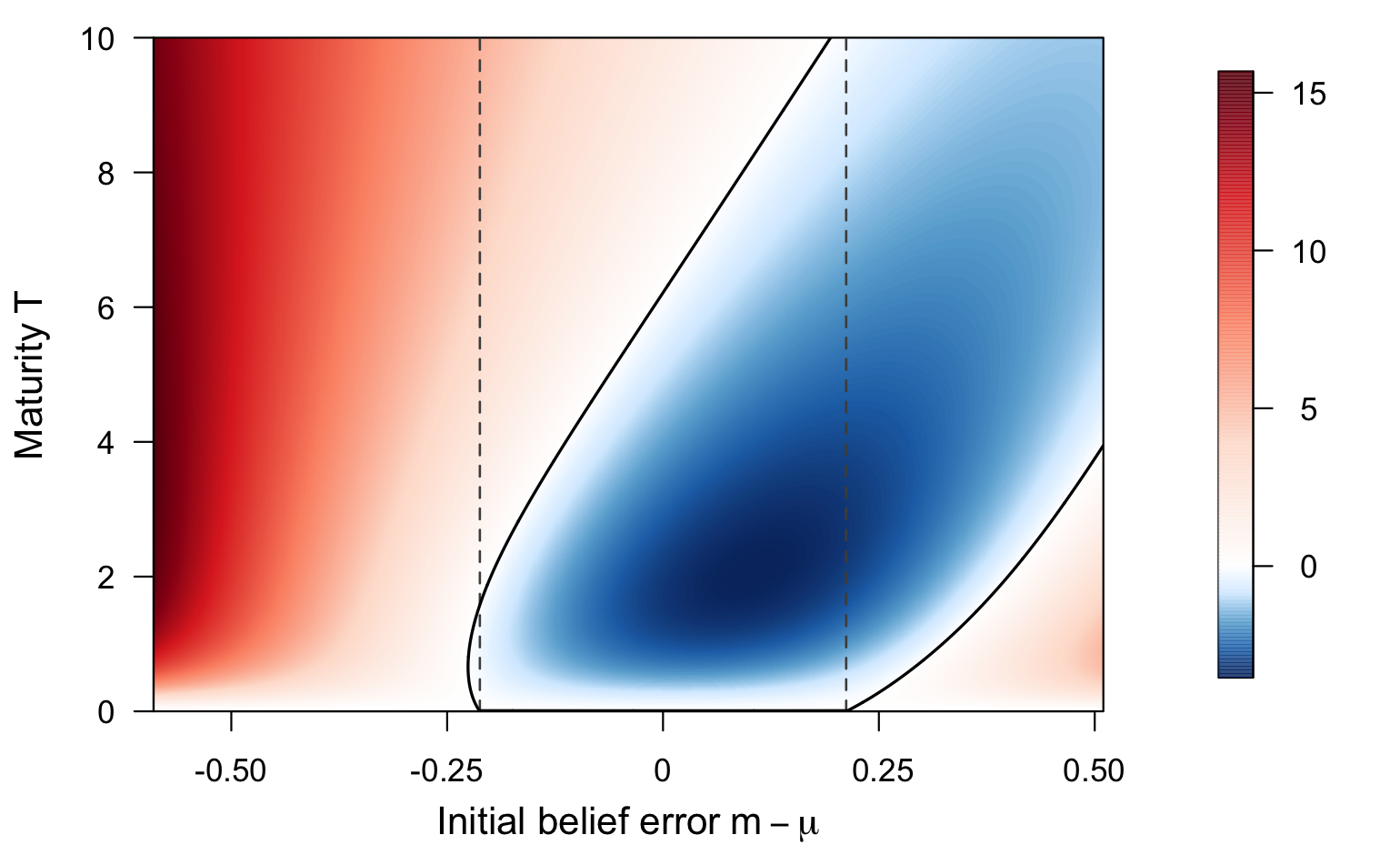}
\caption{Scenario-conditional total value.}
\label{figure_3a}
\end{subfigure}
\hfill
\begin{subfigure}{0.48\linewidth}
\centering
\includegraphics[width=\linewidth]{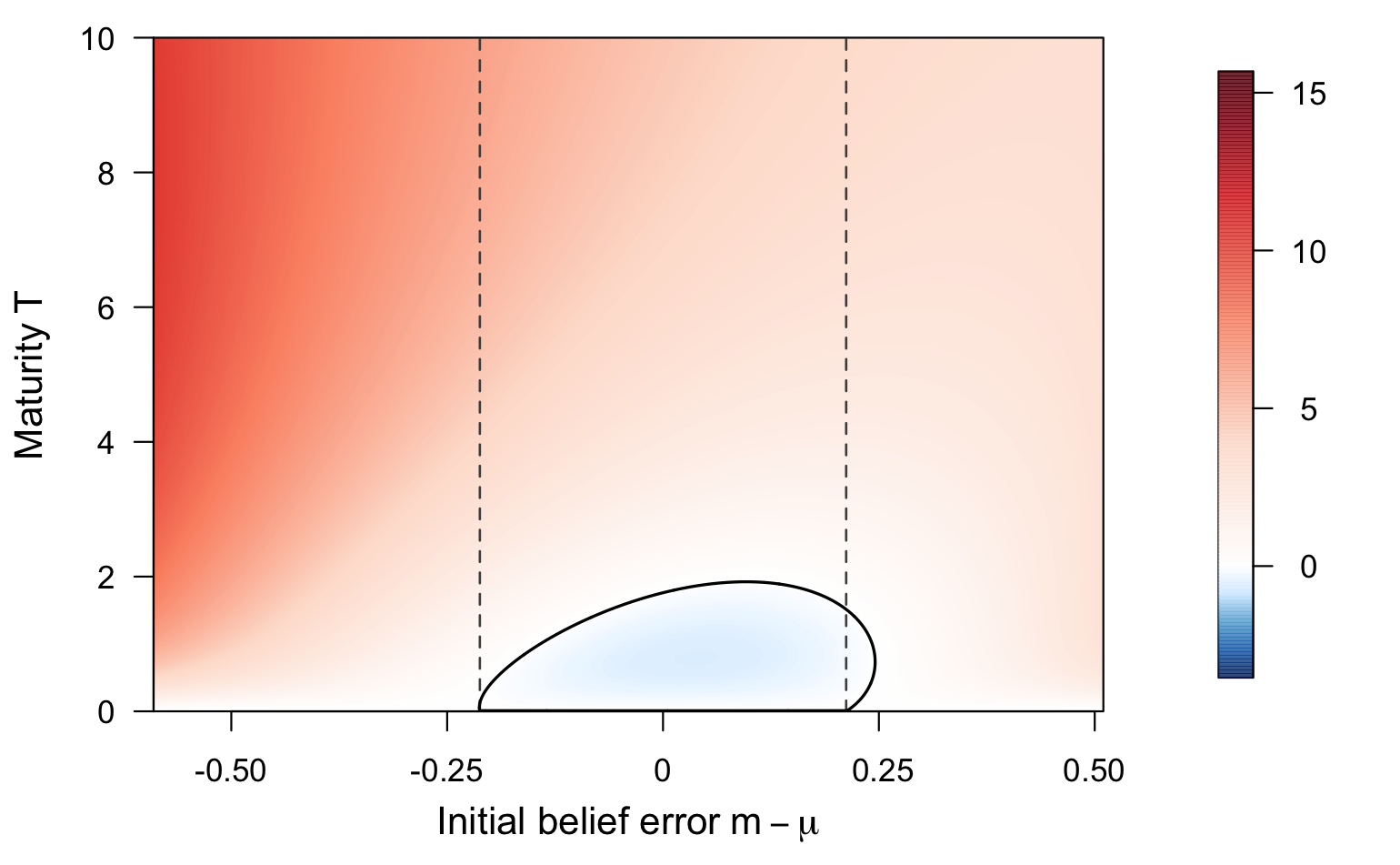}
\caption{Scenario-conditional incremental value.}
\label{figure_3b}
\end{subfigure}
\caption{Scenario-conditional total and incremental value of attention over maturity $T$ and initial belief error $m-\mu$. Panel (a) shows the total value $P^{AI}_{SC}(\mu,\kappa)$ relative to the naive investor, while Panel (b) shows the incremental value $P^{AI}_{SC}(\mu,\kappa)-P^{PO}_{SC}(\mu)$ relative to the price observer. Red (blue) regions indicate positive (negative) values, and the black contours mark zero. The dashed lines delimit the short-horizon reference region $(m-\mu)^2<v/2$ from Proposition \ref{propo_46}. Both panels use the same color scale, and all remaining parameters are fixed at their baseline values.}
\label{short_horizon}
\end{figure}

\subsection{Scenario-Conditional Value across Horizons}

The third question is: how should the investment horizon affect the choice between relying on price observation and adding an auxiliary signal? Figure \ref{short_horizon} separates two performance comparisons. Panel (a) reports total value relative to the fixed-belief strategy, whereas Panel (b) reports incremental value relative to price observation.


Panel (a) shows a broad region in which the scenario-conditional total value of attention is negative. At short maturities, the negative region lies approximately between the dashed lines, consistent with Proposition \ref{propo_46}. Within this region, price observation has negative scenario-conditional value and auxiliary attention reduces it further, so the attentive investor underperforms the naive investor. Outside this region, sufficiently large initial belief errors generate a larger position gap and therefore a stronger correction benefit, producing positive total values.

As maturity increases, the negative-value region in Panel (a) shifts toward positive values of $m-\mu$. This reflects the position-gap mechanism in Theorem \ref{theorem_42}. Rewriting the initial position gap as $\Delta^{NI}=m-\mu-\lb1-q_t^{NI}\rb(m-r)$ shows that its zero is located at $m-\mu=\lb1-q_t^{NI}\rb(m-r)$. Since the baseline posterior risk premium is positive, stronger uncertainty-induced attenuation of the naive investor's investment strategy at longer maturities shifts the zero-gap scenario to the right. The largest negative values occur at relatively short to intermediate maturities around this moving center.

Panel (b) shows that the negative incremental value of auxiliary attention is much more concentrated. It is confined to relatively short maturities and initial belief errors near zero, while the incremental value becomes positive across the displayed range once the horizon is sufficiently long. Near $T=0$, the numerical zero boundary closely follows the dashed short-horizon reference lines. This provides a direct illustration of Proposition \ref{propo_46}: when the initial belief error is small relative to posterior uncertainty, the additional signal generates trading-strategy fluctuations before faster learning has enough time to produce a sufficiently large correction benefit. Larger initial belief errors increase the scope for correction and therefore shorten the horizon required for auxiliary attention to outperform price observation.

The comparison between the panels is particularly informative. Over a substantial region, especially at longer maturities near the center of Panel (a), the incremental value in Panel (b) is positive while the total value remains negative. Auxiliary attention therefore improves scenario-conditional performance relative to price observation without generating enough value to outperform the naive investor. This distinction shows why incremental and total values of information need to be assessed separately.

The decision implication is that an information source cannot be evaluated without specifying both the horizon and the benchmark. Over short horizons, an auxiliary signal can amplify trading fluctuations when the inherited belief is already close to the candidate drift scenario. Over longer horizons in this calibration, the same signal can improve on price observation without generating enough value to outperform fixed-belief choice.

\subsection{Information-Processing Technologies}

The final question is: how does the information-processing technology affect the capacity required to approach perfect learning? We use the price observer, corresponding to $\kappa=0$, as the benchmark and measure convergence by the worst-case normalized gap
\begin{align*}
G_{\rho}(\kappa):=\sup_{\mu\in\mathbb{R}}\frac{P^{AI,\infty}_{SC}(\mu)-P^{AI}_{SC}(\mu,\kappa)}{P^{AI,\infty}_{SC}(\mu)-P^{PO}_{SC}(\mu)}.
\end{align*}
This is the specialization of \eqref{voi_high_capacity_rates} in Proposition \ref{propo_high_capacity_rates} to $\bar\kappa=0$. Under each fixed drift, the numerator is the remaining shortfall of finite-capacity scenario-conditional total value of attention relative to perfect learning, while the denominator is the corresponding shortfall of the price observer. Taking the supremum over $\mu$ selects the drift scenario with the largest normalized remaining shortfall, which motivates the term ``worst-case''. Lower values of $G_{\rho}(\kappa)$ therefore indicate better uniform convergence to perfect learning. In particular, $G_{\rho}(\kappa)=0$ means that no shortfall remains, $0<G_{\rho}(\kappa)<1$ means that attention has reduced the worst-case gap relative to the price observer, and $G_{\rho}(\kappa)>1$ means that the worst-case normalized gap exceeds the price-observer benchmark. A value of $G_{\rho}(\kappa)=1$ means that the worst-case normalized gap equals the price-observer benchmark. It does not imply that $P^{AI}_{SC}(\mu,\kappa)=P^{PO}_{SC}(\mu)$ for every $\mu$, although the corresponding fixed-drift ratio equals one precisely when the two scenario-conditional values coincide in that scenario.

\begin{figure}[htbp] 
\centering 
\begin{subfigure}{0.48\linewidth} 
\centering 
\includegraphics[width=\linewidth]{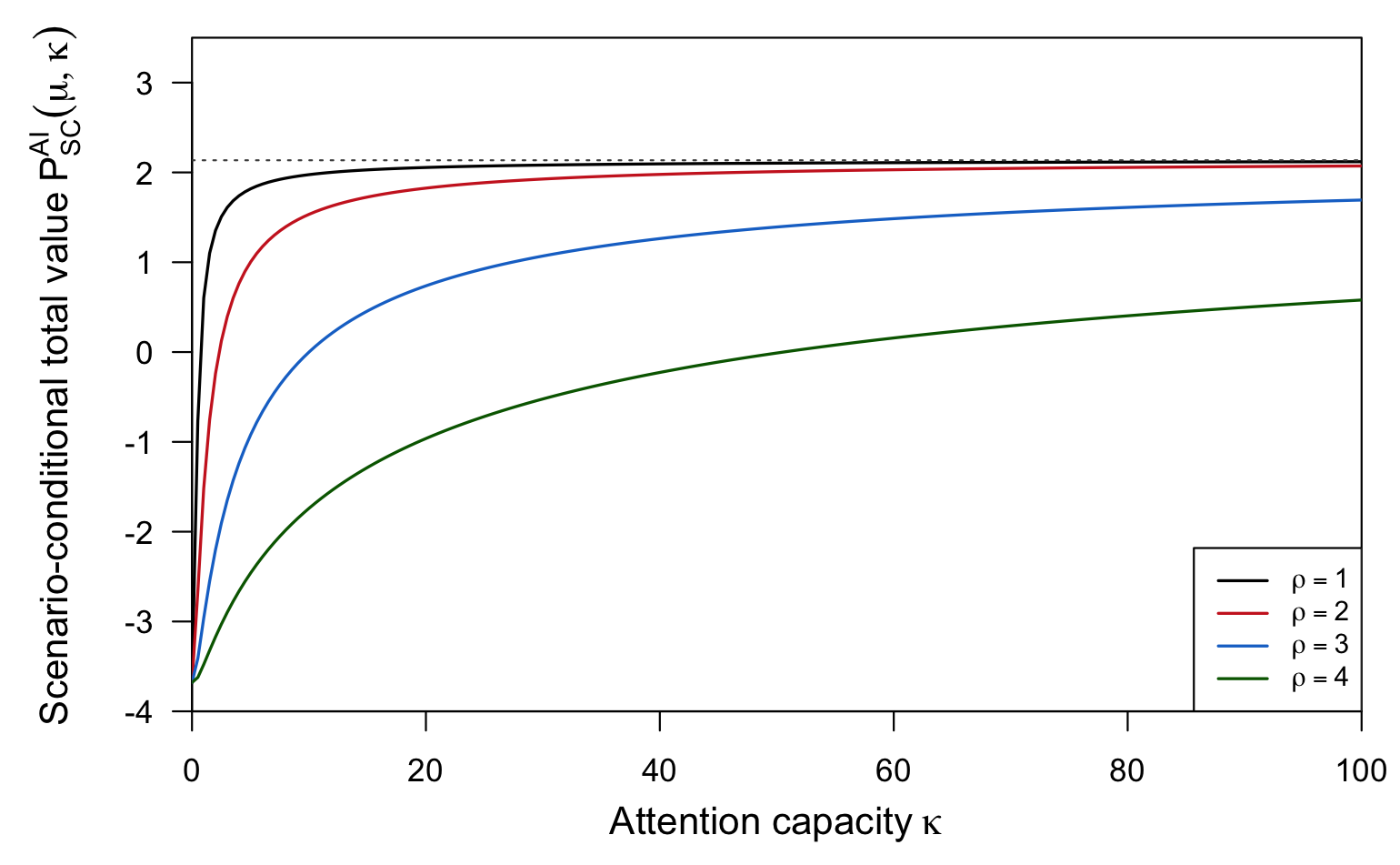} 
\caption{Scenario-conditional value.} \label{figure_4a} 
\end{subfigure} \hfill 
\begin{subfigure}{0.48\linewidth} 
\centering 
\includegraphics[width=\linewidth]{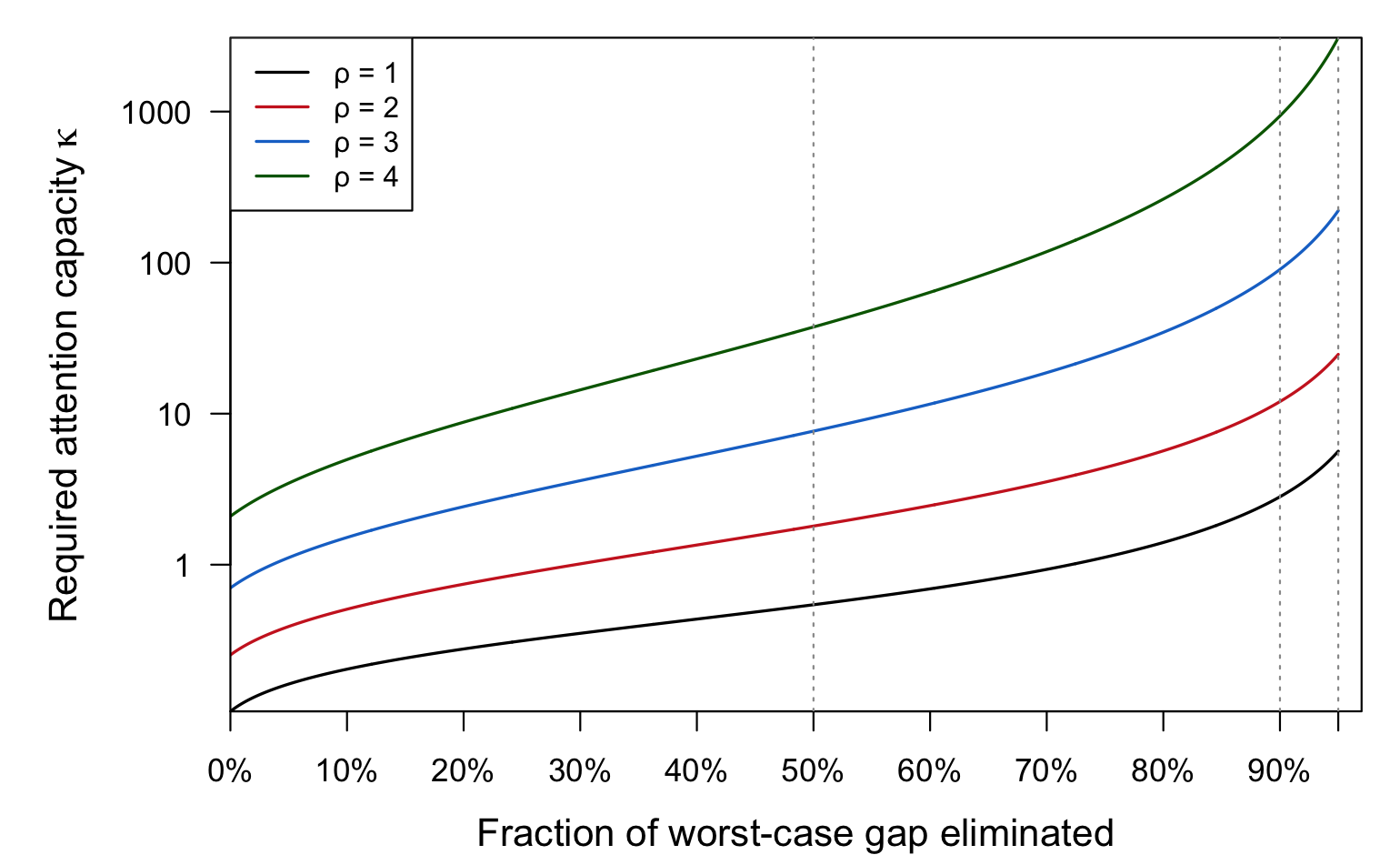} 
\caption{Capacity required for gap elimination.} 
\label{figure_4b} 
\end{subfigure} 
\caption{Scenario-conditional performance under the Shannon and Tsallis capacity indexed by $\rho$. Panel (a) shows the scenario-conditional total value of attention $P^{AI}_{SC}(\mu,\kappa)$ over attention capacity $\kappa$ for $\rho\in\{1,2,3,4\}$. The dotted horizontal line marks the common perfect-learning limit. Panel (b) shows the attention capacity required to eliminate a given fraction of the worst-case normalized gap to perfect learning, with the vertical axis on a logarithmic scale. All remaining parameters are held at their baseline values. The comparison illustrates the high-capacity convergence rates in Proposition \ref{propo_high_capacity_rates}.} 
\label{technology} 
\end{figure}

Figure \ref{technology} compares the technologies indexed by $\rho\in\{1,2,3,4\}$. Panel (a) reports the scenario-conditional total value of attention as capacity increases, while Panel (b) translates $G_{\rho}(\kappa)$ into the capacity required to eliminate a given fraction of the worst-case gap. Table \ref{technology_thresholds} reports the corresponding capacities at selected thresholds.

Panel (a) shows that all four technologies start from the same value at $\kappa=0$. Consistent with \eqref{p_sc_limit}, they also share the same perfect-learning limit, while $\rho$ determines how quickly that limit is approached as capacity increases. The differences are pronounced. The $\rho=1$ technology approaches the perfect-learning limit rapidly, followed by $\rho=2$. Convergence is substantially slower for $\rho=3$ and particularly slow for $\rho=4$. Thus, the same numerical level of attention capacity can correspond to very different distances from perfect learning across technologies.

Panel (b) evaluates this convergence uniformly across fixed-drift scenarios through $G_{\rho}(\kappa)$. Its horizontal axis reports the fraction of the price-observer benchmark gap that has been eliminated, equal to $1-G_{\rho}(\kappa)$ whenever $G_{\rho}(\kappa)\leq1$. Eliminating $50\%$ of the gap means that $G_{\rho}(\kappa)=0.5$, so half of the worst-case shortfall associated with the price observer remains. Eliminating $90\%$ means that $G_{\rho}(\kappa)=0.1$, while eliminating $95\%$ means that only $0.05$ of the benchmark shortfall remains.

The break-even capacity in Table \ref{technology_thresholds} has a distinct interpretation. Starting from $G_{\rho}(0)=1$, the worst-case normalized gap initially rises above $1$ for the technologies considered here. The break-even capacity is the positive value of $\kappa$ at which $G_{\rho}(\kappa)$ returns to $1$. At this point, the worst-case distance from perfect learning is again equal to that of the price observer. Below the break-even point, auxiliary attention can therefore worsen the worst-case normalized gap, while beyond it the numerical curves move below the price-observer benchmark and progressively close the remaining distance to perfect learning. This break-even concept concerns the relative shortfall from perfect learning and should not be confused with a capacity at which the scenario-conditional total value itself equals zero.

\begin{table}[htbp] \centering \caption{Attention capacity required to reduce the worst-case normalized gap to perfect learning. The break-even capacity is the positive capacity at which the worst-case normalized gap returns to its no-attention level.} \label{technology_thresholds} \begin{tabular}{ccrrrr} \toprule $\rho$ & Asymptotic rate $r_\rho(\kappa)$ & Break-even & $50\%$ eliminated & $90\%$ eliminated & $95\%$ eliminated \\ \midrule $1$ & $\kappa^{-1}$ & $0.107$ & $0.543$ & $2.816$ & $5.642$ \\ $2$ & $\kappa^{-1}$ & $0.252$ & $1.802$ & $12.022$ & $24.690$ \\ $3$ & $\ln(1+\kappa)/\kappa$ & $0.701$ & $7.660$ & $90.132$ & $220.892$ \\ $4$ & $\kappa^{-2/3}$ & $2.090$ & $37.463$ & $936.088$ & $3084.616$ \\ \bottomrule \end{tabular} 
\end{table}

Table \ref{technology_thresholds} shows that the required capacities differ sharply across technologies. Under $\rho=1$, the break-even point is reached at $\kappa=0.107$, $50\%$ of the worst-case gap is eliminated at $\kappa=0.543$, $90\%$ at $\kappa=2.816$, and $95\%$ at $\kappa=5.642$. Under $\rho=2$, the corresponding capacities are already considerably larger. The differences become much more pronounced for $\rho=3$ and $\rho=4$. Reaching a $90\%$ reduction requires $\kappa=90.132$ under $\rho=3$ and $\kappa=936.088$ under $\rho=4$, compared with only $\kappa=2.816$ under $\rho=1$. The capacity required for a $95\%$ reduction rises further to $220.892$ and $3084.616$, respectively. Hence, technologies that appear relatively similar at modest target reductions can differ by orders of magnitude when the objective is to approach perfect learning closely.

These differences are the finite-capacity manifestation of the asymptotic rates in Proposition \ref{propo_high_capacity_rates}. Equation \eqref{voi_high_capacity_rates} shows that, at high capacity, $G_{\rho}(\kappa)$ is bounded above and below by constant multiples of $r_{\rho}(\kappa)$. The rate $r_{\rho}(\kappa)$ therefore determines the asymptotic order at which the worst-case normalized gap vanishes. For $\rho=1$ and $\rho=2$, the rate is $\kappa^{-1}$. These technologies have the same asymptotic order even though Table \ref{technology_thresholds} shows substantial differences in their finite-capacity requirements. The rate determines how the remaining gap scales as $\kappa$ becomes large, not the exact capacity needed to reach a particular reduction.

At $\rho=3$, the rate is $\ln(1+\kappa)/\kappa$, which decays more slowly than $\kappa^{-1}$ because of the logarithmic factor. Under $\rho=4$, the rate is $\kappa^{-2/3}$ and convergence is slower still. This rate ordering explains why the capacity requirements in Panel (b) separate increasingly as the targeted fraction of the gap eliminated approaches $100\%$. If the remaining normalized gap is required to be of order $\varepsilon$, the corresponding capacity is of order $\varepsilon^{-1}$ under $\rho=1$ and $\rho=2$, approximately $\varepsilon^{-1}\ln(1/\varepsilon)$ under $\rho=3$, and $\varepsilon^{-3/2}$ under $\rho=4$. Thus, increasingly tight approximation to perfect learning becomes disproportionately more capacity intensive under the slower attention technologies.

Panels (a) and (b) therefore provide complementary perspectives. Panel (a) shows scenario-conditional value under one fixed-drift scenario, while Panel (b) measures the remaining relative shortfall uniformly across all fixed-drift scenarios. The common perfect-learning limit shows that the technologies ultimately lead to the same informational benchmark. Their economically relevant difference lies in the amount of attention capacity required to approach that benchmark.

\section{Conclusion} \label{conclusion}

This paper studies the value of information in dynamic portfolio choice when expected returns are not observed. We distinguish {the classical ex ante value of information under the investor's current belief} from {scenario-conditional value under fixed-drift scenarios}. The latter evaluation holds the trading strategy and filtering rule fixed {without reoptimization for the scenario under consideration} while averaging over future return shocks. In a finite-horizon Gaussian--CARA model, we compare a fixed-belief investor, a price observer, and an attentive investor with an attention capacity-constrained auxiliary signal.

Ex ante, price learning and auxiliary attention are valuable, and the total value of attention increases with capacity. Across fixed-drift scenarios, however, the ranking can reverse because learning corrects an initial trading-position gap while belief revisions generate a fluctuation cost. {We derive the current-belief distribution of scenario-conditional value and the associated underperformance probability. These quantities provide a diagnostic of where information helps or hurts and how much probability the investor's current belief assigns to adverse scenarios, thereby complementing rather than replacing the classical ex ante value of information. Over sufficiently short horizons, the scenario-conditional value of learning is negative when the initial belief error is small relative to posterior uncertainty. At every fixed positive horizon, sufficiently high capacity eventually dominates any fixed finite-capacity benchmark uniformly across drift scenarios. {The worst-case normalized shortfall from perfect learning decreases at rate $\kappa^{-1}$ under Shannon entropy and Tsallis entropy with $1<\rho<3$, at rate $\ln(1+\kappa)/\kappa$ for $\rho=3$, and at rate $\kappa^{-2/(\rho-1)}$ for $\rho>3$.} Information should therefore be evaluated jointly with the trading strategy it induces, the investment horizon, and the information-processing technology.

Several extensions merit further study. First, introducing information-acquisition or processing costs would endogenize attention capacity, following the cost-based approach of \citet{siu_2026}. Allowing the investor to choose among alternative information-processing technologies would be a further extension. Second, robust strategies optimized under drift ambiguity could be evaluated using the same scenario-conditional criterion to assess whether they reduce underperformance in adverse candidate drift scenarios and how any such improvement trades off against ex ante value. Third, transaction costs would make belief-driven portfolio adjustments directly costly, but they could also reduce the frequency and magnitude of these adjustments. How these opposing effects reshape the correction--fluctuation trade-off and the value of additional information remains an open question. We leave these extensions for future research.

\section*{Acknowledgments}
We are grateful to An Chen, Ning Wang, and Xingang Wen for their very helpful comments. Shihao Zhu gratefully acknowledges financial support from the Deutsche Forschungsgemeinschaft (DFG), Project-ID 509303834, FOR 5583 ``Asset Allocation and Asset Pricing under Regulatory Uncertainty''.

\bibliography{bibliography}

\newpage

\appendix

\section{Appendix}\label{appendixa}

\subsection{Instantaneous Auxiliary Information Rate:} \label{derivation_information_rate}

Fix $s\in[t,T)$ and condition on observations up to $s$, so that the current posterior variance is $v_s>0$. Freeze the auxiliary precision at $p_s^\xi$ over $[s,s+\Delta s]$. Standard linear-Gaussian filtering implies the following precision updates \citep{liptser2001statistics}:
\begin{align*}
\frac{1}{v_{s+\Delta s}}
&=\frac{1}{v_s}+\lb\frac{1}{\sigma^2}+p_s^\xi\rb\Delta s,\\
\frac{1}{v^P_{s+\Delta s}}
&=\frac{1}{v_s}+\frac{\Delta s}{\sigma^2}. \labeleq{posterior_precision_increment}
\end{align*}
It follows from \eqref{posterior_precision_increment} that
\begin{align*}
v^P_{s+\Delta s}-v_{s+\Delta s}
=p_s^\xi v_s^2\Delta s+o(\Delta s). \labeleq{posterior_variance_increment}
\end{align*}
Because both posterior variances converge to $v_s$, the mean value theorem,
and \eqref{posterior_variance_increment} yield
\begin{align*}
\lim_{\Delta s\downarrow0}
\frac{\mathcal{H}(v^P_{s+\Delta s})-\mathcal{H}(v_{s+\Delta s})}{\Delta s}
=p_s^\xi v_s^2\mathcal{H}'(v_s),
\end{align*}
which proves \eqref{instant_information_rate}.

\subsection{Correction Benefit, Fluctuation Cost, Performance Threshold and Distribution for the Price Observer:}\label{special_case_po}

Setting $\kappa=0$ and substituting the explicit posterior variance from \eqref{v_po_explicit} into the ODE system \eqref{ode_system_correction_fluctuation} yields
\begin{align*}
B(t;0)=\frac{\lb1-q_t^{NI}\rb a(t,0)}{2-q_t^{NI}},\qquad C(t;0)=\frac{1-q_t^{NI}}{2}-\frac{1}{2}\ln\lb2-q_t^{NI}\rb. \labeleq{b_c_po}
\end{align*}
Substituting both into the performance threshold for $\kappa=0$ yields
\begin{align*}
h(t,0)=\sqrt{\frac{\lb2-q_t^{NI}\rb\left(1-q_t^{NI}-\ln\lb2-q_t^{NI}\rb\right)}{2\lb1-q_t^{NI}\rb a(t,0)}}.
\end{align*}
With Proposition \ref{propo_43}, this yields 
\begin{align*}
P_{SC}^{PO}(\mu)\overset{d}{=}\frac{e^{-r(T-t)}}{\gamma}\left[\frac{\lb1-q_t^{NI}\rb a(t,0)}{2-q_t^{NI}}Z_t^2-\frac{1-q_t^{NI}}{2}+\frac{1}{2}\ln\lb2-q_t^{NI}\rb\right],
\end{align*}
where
\begin{align*}
Z_t\sim\N\lb-\lb1-q_t^{NI}\rb(m-r),v\rb.
\end{align*}

\subsection{Convergence of $B(t;\kappa)$ and $C(t;\kappa)$:} \label{convergence_c_b}
\begin{lemma}\label{lemma_a1}
Let $(t,x,m,v)\in[0,T)\times\mathcal{O}$ be fixed. Then
\begin{align*}
\lim_{\kappa\to\infty}B(t;\kappa)=a(t,0), \qquad\lim_{\kappa\to\infty}C(t;\kappa)=0.
\end{align*}
\end{lemma}

\begin{proof}
Write $v_s^\kappa:=v_s^{\mathrm{AI}}$ and $q_s^\kappa:=q_s^{\mathrm{AI}}$, and define
\begin{align*}
A_\kappa&:=a(t,0)-B(t;\kappa),\\
E_\kappa&:=\frac{1}{2\sigma^2}\int_t^T q_s^\kappa v_s^\kappa ds,\\
w_0&:=q_t^{\mathrm{NI}}v>0.
\end{align*}
By \eqref{b_capacity_representation} and \eqref{b_ai_infty},
\begin{align*}
E_\kappa=b^{\mathrm{AI},\infty}(t,v)-b^{\mathrm{AI}}(t,v;\kappa).
\end{align*}
Moreover, \eqref{ex_ante_attention_value_aggreagated} and the identity
\begin{align*}
1+2vB(t;\kappa)=\lb 1+\frac{(T-t)v}{\sigma^2}\rb(1-2w_0A_\kappa)
\end{align*}
yield
\begin{align*}
E_\kappa=C(t;\kappa)-\frac{1}{2}\ln(1-2w_0A_\kappa).
\end{align*}
Since $0<A_\kappa<a(t,0)$ and $C(t;\kappa)>0$ by Theorem \ref{theorem_42}, we obtain
\begin{align*}
0<C(t;\kappa)&\le E_\kappa,\\
0<A_\kappa&\le\frac{1-e^{-2E_\kappa}}{2w_0} \le\frac{E_\kappa}{w_0}.
\end{align*}
By \eqref{variance_perfect_learning}, $v_s^\kappa\to0$ as $\kappa\to\infty$ for every $s>t$. Furthermore,
\begin{align*}
0<q_s^\kappa v_s^\kappa\le v, \qquad s\in[t,T].
\end{align*}
Dominated convergence therefore implies $E_\kappa\to0$. The preceding bounds give
\begin{align*}
A_\kappa\to0, \qquad C(t;\kappa)\to0,
\end{align*}
which proves both assertions.
\end{proof}

\subsection{Auxiliary Technical Lemma: }

\begin{lemma}\label{lemma_a3_aux}
Let $(t,x,m,v)\in[0,T)\times\mathcal{O}$, $\mu\in\mathbb{R}$, and $\kappa\in[0,\infty)$ be fixed. Let $\widetilde{\PP}^{\mu}$ be the probability measure satisfying
\begin{align*}
\left.\frac{d\widetilde{\PP}^{\mu}}{d\PP^\mu}\right|_{\mathcal{F}^{AI}_T}
=\exp\lb-\int_t^T\frac{q_s^{AI}(m_s^{AI}-r)}{\sigma}dW_s^1-\frac{1}{2}\int_t^T\lb\frac{q_s^{AI}(m_s^{AI}-r)}{\sigma}\rb^2ds\rb. \labeleq{tilted_measure}
\end{align*}
Then the attentive investor's scenario-conditional expected utility satisfies
\begin{align*}
J^{AI}(t,x,m,v;\mu,\kappa)=V^{FI}(t,x;\mu)\widetilde{\E}^{\mu}_{t,x,m,v}\lsb\exp\lb\frac{1}{2\sigma^2}\int_t^T\lb\Delta_s^{AI}\rb^2ds\rb\rsb, \labeleq{j_ai_presentation}
\end{align*}
where 
\begin{align*}
    \Delta_s^{AI}&:=q_s^{AI}(m_s^{AI}-r)-(\mu-r), \qquad s\in[t,T]. \labeleq{delta_ai}
\end{align*}

\end{lemma}
\begin{proof}
 
For fixed $\kappa$, set
\begin{align*}
h_s:=\frac{q_s^{AI}(m_s^{AI}-r)}{\sigma},
\qquad
\theta:=\frac{\mu-r}{\sigma}.
\end{align*}
Because $h$ is affine in the linear-Gaussian state with bounded deterministic coefficients, the Bene\v{s} linear-growth criterion \citep{klebaner2014stochastic} shows that the stochastic exponential in \eqref{tilted_measure} is a true martingale. Under $\widetilde{\PP}^{\mu}$, completing the square gives
\begin{align*}
J^{AI}
={}&-\frac{1}{\gamma}e^{-\gamma e^{r(T-t)}x-\frac{1}{2}\theta^2(T-t)}
\widetilde{\E}^{\mu}_{t,x,m,v}
\lsb\exp\lb\frac{1}{2}\int_t^T(h_s-\theta)^2ds\rb\rsb\\
={}&V^{FI}(t,x;\mu)
\widetilde{\E}^{\mu}_{t,x,m,v}
\lsb\exp\lb\frac{1}{2\sigma^2}
\int_t^T\lb\Delta_s^{AI}\rb^2ds\rb\rsb,
\end{align*}
which proves \eqref{j_ai_presentation}. It remains to verify finiteness of the exponential moment in \eqref{j_ai_presentation}. Under each $\mathbb{P}^{\vartheta}$, $m^{AI}$ is a Gaussian process whose covariance is independent of $\vartheta$. The parameter $\vartheta$ changes only its mean. Thus, by \eqref{scenario_conditional_wealth} and \eqref{optimal_strategy}, 
$\mathbb{E}^{\vartheta}_{t,x,m,v}[\exp(-\gamma X_T^{\pi^{*,AI}})]$
is a Gaussian quadratic exponential with a common quadratic kernel. Its prior average equals the finite Bayesian moment $-\gamma V^{AI}(t,x,m,v;\kappa)$, so it is finite for almost every $\vartheta$ under the prior. Since Gaussian quadratic integrability depends only on the common quadratic kernel, it is finite for every $\vartheta$, in particular for $\mu$. Hence $J^{AI}$, and therefore the exponential moment in
\eqref{j_ai_presentation}, is finite.

\end{proof}

\par\vspace{\baselineskip}
\subsection{Proof of Proposition \ref{proposition_31}:}\label{proof_proposition_31}
\begin{proof}

For use throughout the appendix, we define the total posterior-variance decay rate by
\begin{align*}
G_\kappa(y):=\frac{y^2}{\sigma^2}+\frac{\kappa y^2}{f(y)},\qquad y>0. \labeleq{variance_decay_rate}
\end{align*}
Thus, the variance ODE in \eqref{filter_ai_compact} is
$dv_s^{AI}=-G_\kappa(v_s^{AI})ds$. Using
\eqref{wealth_learning_investor} and \eqref{filter_ai_compact}, define for a
smooth function $w=w(s,x,z,y)$ the second-order differential operator
\begin{align*}
\mathcal{L}^\pi w
:={}&\lb rx+\pi(z-r)\rb w_x-G_\kappa(y)w_y
+\frac{1}{2}\pi^2\sigma^2w_{xx}+\pi yw_{xz}
+\frac{1}{2}G_\kappa(y)w_{zz}.
\end{align*}
The dynamic programming principle for \eqref{portfolio_problem} leads to the HJB equation
\begin{align*}
0=W_s+\sup_{\pi\in\mathbb{R}}\mathcal{L}^\pi W,
\qquad W(T,x,z,y)=U(x). \labeleq{hjb_equation}
\end{align*}


For each $(s,y)\in[0,T)\times\mathbb{R}_+$, let $(v_u)_{u\in[s,T]}$ be the
deterministic solution of
\begin{align*}
dv_u=-G_\kappa(v_u)du,\qquad v_s=y. \labeleq{variance_characteristic}
\end{align*}
Condition \eqref{entropy_regular} implies that $G_\kappa(u)/u$ is bounded on
$(0,y]$, so comparison gives a unique strictly positive solution on $[s,T]$.
Since $G_\kappa\in C^1(0,\infty)$, the solution is continuously
differentiable in $(s,y)$. Define
\begin{align*}
a(s,y)&:=\frac{T-s}{2\lb\sigma^2+(T-s)y\rb}, \labeleq{candidate_a}\\
b^{AI}(s,y;\kappa)&:=\int_s^T G_\kappa(v_u)a(u,v_u)du. \labeleq{candidate_b}
\end{align*}
and consider
\begin{align*}
W(s,x,z,y):=-\frac{1}{\gamma}
\exp\lb-\gamma e^{r(T-s)}x-a(s,y)(z-r)^2-b^{AI}(s,y;\kappa)\rb. \labeleq{candidate_W}
\end{align*}
This function has the required interior regularity and satisfies the terminal
condition.

Since $W_{xx}<0$, its unique pointwise maximizer is
\begin{align*}
\pi^\ast(s,z,y)
&=-\frac{(z-r)W_x+yW_{xz}}{\sigma^2W_{xx}}\\
&=e^{-r(T-s)}\frac{z-r}{\gamma\sigma^2}
\lb1-2a(s,y)y\rb
=e^{-r(T-s)}\frac{z-r}{\gamma\sigma^2}
\frac{\sigma^2}{\sigma^2+(T-s)y}. \labeleq{feedback_candidate}
\end{align*}
Direct differentiation of \eqref{candidate_a}--\eqref{candidate_b} along
\eqref{variance_characteristic} gives
\begin{align*}
a_s(s,y)-G_\kappa(y)a_y(s,y)
&=-\frac{1}{2\sigma^2}\lb1-2a(s,y)y\rb^2
+2G_\kappa(y)a(s,y)^2,\\
b_s^{AI}(s,y;\kappa)-G_\kappa(y)b_y^{AI}(s,y;\kappa)
&=-G_\kappa(y)a(s,y). \labeleq{candidate_coefficient_equations}
\end{align*}
Both functions vanish at $s=T$, and substitution of
\eqref{candidate_coefficient_equations} verifies
\eqref{hjb_equation}. Evaluating \eqref{feedback_candidate} at
$(z,y)=(m_s^{AI},v_s^{AI})$ gives \eqref{optimal_strategy}.

Set
\begin{align*}
q_s:=\frac{\sigma^2}{\sigma^2+(T-s)v_s^{AI}}\in(0,1]. \labeleq{candidate_attenuation}
\end{align*}
By It\^o's isometry and the variance ODE in \eqref{filter_ai_compact},
\begin{align*}
\E^{AI}_{t,x,m,v}\lsb\lb m_s^{AI}-r\rb^2\rsb
=(m-r)^2+\int_t^sG_\kappa(v_u^{AI})du
=(m-r)^2+v-v_s^{AI}. \labeleq{posterior_second_moment}
\end{align*}
Thus, \eqref{feedback_candidate}, \eqref{candidate_attenuation}, and
\eqref{posterior_second_moment} show that the candidate trading strategy is predictable
and satisfies the square-integrability requirement
\eqref{admissible_square_integrability}:
\begin{align*}
\E^{AI}_{t,x,m,v}\lsb\int_t^T|\pi_s^{\ast,AI}|^2ds\rsb<\infty.
\end{align*}

To verify the class-$(D)$ condition, let
\begin{align*}
Y_s:=e^{r(T-s)}X_s^{\pi^{\ast,AI},AI},
\qquad \Xi_s^\ast:=e^{-\gamma Y_s}.
\end{align*}
The discounted-to-maturity wealth satisfies
\begin{align*}
Y_s
={}&e^{r(T-t)}x
+\frac{1}{\gamma\sigma^2}\int_t^sq_u\lb m_u^{AI}-r\rb^2du
+\frac{1}{\gamma\sigma}\int_t^sq_u\lb m_u^{AI}-r\rb
d\widehat W_u^{1,AI}. \labeleq{candidate_wealth_to_maturity}
\end{align*}
Therefore, by \eqref{candidate_wealth_to_maturity},
\begin{align*}
\lb \Xi_s^\ast\rb^2
={}&\lb \Xi_t^\ast\rb^2
\mathcal{E}\lb-\frac{2}{\sigma}\int_t^\cdot
q_u\lb m_u^{AI}-r\rb d\widehat W_u^{1,AI}\rb_s\\
&\times\exp\lb-\frac{2}{\sigma^2}\int_t^s
q_u\lb1-q_u\rb\lb m_u^{AI}-r\rb^2du\rb, \labeleq{candidate_exponential_square}
\end{align*}
where $\mathcal{E}(\cdot)$ denotes the stochastic exponential. In \eqref{candidate_exponential_square}, the stochastic exponential is a
positive local martingale, hence a
supermartingale, while the second factor is at most one. Thus, for all
$\mathbb{F}^{AI}$-stopping times $\tau\in[t,T]$,
\begin{align*}
\sup_{\tau}\E^{AI}_{t,x,m,v}\lsb\lb \Xi_\tau^\ast\rb^2\rsb
\leq\lb \Xi_t^\ast\rb^2<\infty. \labeleq{candidate_class_D_bound}
\end{align*}
By \eqref{candidate_class_D_bound}, the family $(\Xi_\tau^\ast)_\tau$ is
uniformly integrable, and
$\pi^{\ast,AI}\in\mathcal{A}^{AI}(t,x,m,v)$.

For an arbitrary admissible trading strategy $\pi$, the stochastic integral obtained by applying It\^o's formula to the candidate along the controlled state is initially only a local martingale. Let $(\tau_n)_{n\geq1}$ localize this integral. The HJB inequality then gives
\begin{align*}
\E^{AI}_{t,x,m,v}\lsb
W\lb\tau_n,X_{\tau_n}^{\pi,AI},m_{\tau_n}^{AI},v_{\tau_n}^{AI}\rb\rsb
\leq W(t,x,m,v). \labeleq{localized_verification_inequality}
\end{align*}
Because $a,b^{AI}\geq0$,
\begin{align*}
\left|W\lb s,X_s^{\pi,AI},m_s^{AI},v_s^{AI}\rb\right|
\leq\frac{1}{\gamma}
\exp\lb-\gamma e^{r(T-s)}X_s^{\pi,AI}\rb. \labeleq{verification_domination}
\end{align*}
The class-$(D)$ condition \eqref{admissible_class_D}, together with \eqref{localized_verification_inequality}--\eqref{verification_domination}, permits passage to the limit $\tau_n\uparrow T$, yielding
\begin{align*}
\E^{AI}_{t,x,m,v}\lsb U\lb X_T^{\pi,AI}\rb\rsb\leq W(t,x,m,v).
\end{align*}
For $\pi=\pi^{\ast,AI}$, the HJB drift is zero and \eqref{candidate_class_D_bound} gives the required class-$(D)$ property, hence equality. Therefore $W=V^{AI}$, and the unique HJB maximizer \eqref{feedback_candidate} is the optimal feedback trading strategy.
\end{proof}

\subsection{Proof of Proposition \ref{propo_32}:} \label{proof_propo_32}

\begin{proof}
It is sufficient to prove the claim for $b^{AI}$. Write $v_s^\kappa$ for the variance path at capacity $\kappa$. If $0\leq\kappa_1<\kappa_2$, then $G_{\kappa_2}(y)>G_{\kappa_1}(y)$ for every $y>0$ because $f(y)>0$. Both paths start from $v$ and remain strictly positive by \eqref{entropy_regular}. The strict scalar ODE comparison theorem therefore
gives
\begin{align*}
v_s^{\kappa_2}<v_s^{\kappa_1},\qquad s\in(t,T]. \labeleq{variance_capacity_order}
\end{align*}

We next rewrite $b^{AI}$ in a form that preserves this ordering. The variance
ODE in \eqref{filter_ai_compact} implies
\begin{align*}
d\ln\lb\sigma^2+(T-s)v_s^\kappa\rb
=-\frac{v_s^\kappa+(T-s)G_\kappa(v_s^\kappa)}
{\sigma^2+(T-s)v_s^\kappa}ds.
\end{align*}
Consequently, by \eqref{candidate_a},
\begin{align*}
G_\kappa(v_s^\kappa)a(s,v_s^\kappa)ds
={}&-\frac{1}{2}d\ln\lb\sigma^2+(T-s)v_s^\kappa\rb
-\frac{1}{2}\frac{v_s^\kappa}
{\sigma^2+(T-s)v_s^\kappa}ds.
\end{align*}
Since $v_t^\kappa=v$ and the logarithm has argument $\sigma^2$ at $s=T$,
integration and \eqref{b_integral} yield
\begin{align*}
b^{AI}(t,v;\kappa)
=\frac{1}{2}\ln\lb\frac{\sigma^2+(T-t)v}{\sigma^2}\rb
-\frac{1}{2}\int_t^T
\frac{v_s^\kappa}{\sigma^2+(T-s)v_s^\kappa}ds. \labeleq{b_capacity_representation}
\end{align*}
For every $s\in[t,T]$,
\begin{align*}
\frac{\partial}{\partial y}\frac{y}{\sigma^2+(T-s)y}
=\frac{\sigma^2}{\lb\sigma^2+(T-s)y\rb^2}>0.
\end{align*}
Thus, \eqref{variance_capacity_order} and the strict monotonicity of
$y/(\sigma^2+(T-s)y)$ imply
\begin{align*}
\int_t^T\frac{v_s^{\kappa_2}}
{\sigma^2+(T-s)v_s^{\kappa_2}}ds
<\int_t^T\frac{v_s^{\kappa_1}}
{\sigma^2+(T-s)v_s^{\kappa_1}}ds.
\end{align*}
Together with \eqref{b_capacity_representation}, this gives
$b^{AI}(t,v;\kappa_2)>b^{AI}(t,v;\kappa_1)$. Finally,
\eqref{value_function} is strictly increasing in $b^{AI}$, which proves the
claim for $V^{AI}$.
\end{proof}

\subsection{Proof of Proposition \ref{propo_34}:} \label{propo_34_proof}
\begin{proof}
By \eqref{investor_filtrations}, conditional on $\mathcal{G}_t$, every
$\mathbb{F}^{NI}$-predictable square-integrable trading strategy is a deterministic
element of $L^2([t,T])$. Define
\begin{align*}
k_s:=e^{r(T-s)}\pi_s.
\end{align*}
By the conditional-independence assumption following \eqref{prior_belief} and
the wealth equation \eqref{wealth_naive_investor}, terminal wealth is Gaussian
with mean and variance
\begin{align*}
\E^{NI}_{t,x,m,v}\lsb X_T^\pi\rsb
&=e^{r(T-t)}x+(m-r)\int_t^Tk_sds,\\
\Var^{NI}_{t,x,m,v}\lb X_T^\pi\rb
&=v\lb\int_t^Tk_sds\rb^2+\sigma^2\int_t^Tk_s^2ds. \labeleq{naive_terminal_moments}
\end{align*}
By the CARA utility \eqref{cara_utility}, the Gaussian moment-generating
function reduces the problem to
minimizing over $k\in L^2([t,T])$ the control-dependent exponent
\begin{align*}
{}&-\gamma(m-r)\int_t^Tk_sds
+\frac{\gamma^2v}{2}\lb\int_t^Tk_sds\rb^2
+\frac{\gamma^2\sigma^2}{2}\int_t^Tk_s^2ds. \labeleq{naive_control_exponent}
\end{align*}
For a fixed value of $\int_t^Tk_sds$, the first two terms are fixed, while
Cauchy--Schwarz gives
\begin{align*}
\int_t^Tk_s^2ds\geq
\frac{1}{T-t}\lb\int_t^Tk_sds\rb^2,
\end{align*}
with equality if and only if $k$ is constant almost everywhere. Hence every
minimizer must satisfy $k_s^\ast=c$. Substitution reduces the problem to minimizing
the strictly convex scalar expression
\begin{align*}
(T-t)\lsb-\gamma(m-r)c
+\frac{\gamma^2}{2}\lb\sigma^2+(T-t)v\rb c^2\rsb.
\end{align*}
Its first-order condition yields
\begin{align*}
c=\frac{m-r}{\gamma\lb\sigma^2+(T-t)v\rb},
\end{align*}
which is the unique global minimizer. Therefore,
\begin{align*}
\pi_s^{\ast,NI}
=e^{-r(T-s)}\frac{m-r}{\gamma\sigma^2}
\frac{\sigma^2}{\sigma^2+(T-t)v}.
\end{align*}
This is \eqref{strategy_naive_investor}.

The candidate satisfies \eqref{admissible_square_integrability} and is therefore
admissible. More generally, \eqref{naive_terminal_moments} shows that terminal
wealth is Gaussian under every $L^2$ trading strategy, so its terminal exponential
moment is finite. Finally, substituting the optimizer into
\eqref{naive_control_exponent} gives
\begin{align*}
{}&-\gamma(m-r)\int_t^Tk_s^\ast ds
+\frac{\gamma^2v}{2}\lb\int_t^Tk_s^\ast ds\rb^2
+\frac{\gamma^2\sigma^2}{2}\int_t^T\lb k_s^\ast\rb^2ds\\
&=-\frac{(T-t)(m-r)^2}{2\lb\sigma^2+(T-t)v\rb}
=-a(t,v)(m-r)^2,
\end{align*}
which gives \eqref{value_function_naive_investor}.
\end{proof}

\subsection{Proof of Proposition \ref{propo_35}:} \label{propo_35_proof}
\begin{proof}
Write $v_s^\kappa:=v_s^{AI}$. Since $G_\kappa>0$, the variance path is
decreasing. For $\kappa>0$ and $s>t$, separation of the variance ODE in
\eqref{filter_ai_compact} gives
\begin{align*}
s-t
=\int_{v_s^\kappa}^v\frac{du}{G_\kappa(u)}
\leq\frac{1}{\kappa}\int_{v_s^\kappa}^v\frac{f(u)}{u^2}du.
\end{align*}
The inequality uses $G_\kappa(u)\geq\kappa u^2/f(u)$. Fix
$\varepsilon\in(0,v)$. If $v_s^\kappa\geq\varepsilon$, then
\begin{align*}
\kappa(s-t)\leq\int_\varepsilon^v\frac{f(u)}{u^2}du<\infty,
\end{align*}
which is impossible for all sufficiently large $\kappa$. Hence
\begin{align*}
v_s^\kappa\longrightarrow0,\qquad s\in(t,T]. \labeleq{variance_perfect_learning}
\end{align*}

In \eqref{b_capacity_representation}, the integrand is bounded by
$v/\sigma^2$ and converges to zero for almost every $s$ by
\eqref{variance_perfect_learning}. Dominated convergence therefore yields
\begin{align*}
b^{AI}(t,v;\kappa)\longrightarrow
b^{AI,\infty}(t,v)
=\frac{1}{2}\ln\lb1+\frac{(T-t)v}{\sigma^2}\rb. \labeleq{b_perfect_limit}
\end{align*}
The limit agrees with \eqref{b_ai_infty}.
Substitution into \eqref{value_function} proves
\eqref{value_function_ai_infty_explicit}.

For completeness, the standard one-state CARA problem with known drift has the trading strategy and value stated in \eqref{full_information_strategy} and
\eqref{value_function_full_info}:
\begin{align*}
\pi_s^{\ast,FI}
&=e^{-r(T-s)}\frac{\mu-r}{\gamma\sigma^2},
\qquad s\in[t,T],\\
V^{FI}(t,x;\mu)
&=-\frac{1}{\gamma}\exp\lb-\gamma e^{r(T-t)}x
-\frac{T-t}{2\sigma^2}(\mu-r)^2\rb.
\end{align*}
If $\mu\sim\N(m,v)$, Gaussian integration gives
\begin{align*}
\E_{m,v}\lsb
e^{-\frac{T-t}{2\sigma^2}(\mu-r)^2}\rsb
&=\frac{
\exp\lb-\frac{(T-t)(m-r)^2}
{2\lb\sigma^2+(T-t)v\rb}\rb}
{\sqrt{1+\frac{(T-t)v}{\sigma^2}}}\\
&=\exp\lb-a(t,v)(m-r)^2-b^{AI,\infty}(t,v)\rb.
\end{align*}
The last equality follows from \eqref{a_explicit} and
\eqref{b_perfect_limit}. Multiplying by the wealth-dependent factor in
$V^{FI}$ gives
\begin{align*}
\E_{m,v}\lsb V^{FI}(t,x;\mu)\rsb
&=-\frac{1}{\gamma}\exp\lb-\gamma e^{r(T-t)}x
-a(t,v)(m-r)^2-b^{AI,\infty}(t,v)\rb\\
&=V^{AI,\infty}(t,x,m,v).
\end{align*}

Finally, let
\begin{align*}
q_s^\kappa:=\frac{\sigma^2}{\sigma^2+(T-s)v_s^\kappa}.
\end{align*}
By \eqref{posterior_moments}, $m_s^{AI}=\E[\mu\mid\mathcal{F}_s^{AI}]$ and
$v_s^\kappa=\Var(\mu\mid\mathcal{F}_s^{AI})$ is deterministic. The tower
property gives
\begin{align*}
\E^{AI}_{t,x,m,v}\lsb\lb m_s^{AI}-\mu\rb^2\rsb
&=\E^{AI}_{t,x,m,v}\lsb\Var\lb\mu\mid\mathcal{F}_s^{AI}\rb\rsb
=v_s^\kappa. \labeleq{posterior_mse_bayesian}
\end{align*}
Thus, by \eqref{variance_perfect_learning}--\eqref{posterior_mse_bayesian}, for
every $s>t$, $m_s^{AI}\to\mu$ in $L^2$ and
$q_s^\kappa\to1$. Moreover,
\begin{align*}
{}&\E^{AI}_{t,x,m,v}\lsb
\lb q_s^\kappa(m_s^{AI}-r)-(\mu-r)\rb^2\rsb\\
&\leq2v_s^\kappa
+2\lb1-q_s^\kappa\rb^2\lsb(m-r)^2+v\rsb. \labeleq{strategy_error_bound}
\end{align*}
The right-hand side of \eqref{strategy_error_bound} converges to zero for every
$s>t$ and is bounded by
an integrable constant independent of $s$ and $\kappa$. Dominated convergence
therefore also gives convergence in $L^2(\PP\otimes ds)$:
\begin{align*}
\pi_s^{\ast,AI}
\longrightarrow e^{-r(T-s)}\frac{\mu-r}{\gamma\sigma^2}.
\end{align*}
This proves \eqref{strategy_ai_infty}. The value assigned at time $s=t$ is immaterial.
\end{proof}

\subsection{Proof of Theorem \ref{theorem_42}:}\label{proof_theorem_42}

\begin{proof}
Under $\widetilde{\PP}^{\mu}$ defined by \eqref{tilted_measure},
Girsanov's theorem gives
\begin{align*}
dm_s^{AI}={}&-\left[\frac{G_\kappa(v_s^{AI})}{v_s^{AI}}\left(m_s^{AI}-\mu\right)+\frac{v_s^{AI}q_s^{AI}}{\sigma^2}\left(m_s^{AI}-r\right)\right]ds\\
&+\frac{v_s^{AI}}{\sigma}d\widetilde W_s^1+v_s^{AI}\sqrt{\frac{\kappa}{f(v_s^{AI})}}dW_s^2. \labeleq{tilted_filter_dynamics}
\end{align*}
Using \eqref{q_ai} and \eqref{filter_ai_compact}, direct differentiation gives
\begin{align*}
\frac{d}{ds}q_s^{AI}=q_s^{AI}v_s^{AI}\left[\left(1-q_s^{AI}\right)\left(\frac{1}{\sigma^2}+\frac{\kappa}{f(v_s^{AI})}\right)+\frac{q_s^{AI}}{\sigma^2}\right].
\end{align*}
Hence, using the definition of $\Delta_s^{AI}$ in \eqref{delta_ai},
\begin{align*}
d\Delta_s^{AI}=-\lambda_s^\kappa\Delta_s^{AI}ds+\sqrt{q_s^{AI}v_s^{AI}\lambda_s^\kappa} d\widetilde W_s, \qquad \Delta_t^{AI}=\Delta^{NI}, \labeleq{tilted_gap_dynamics}
\end{align*}
for a $\widetilde{\PP}^{\mu}$-Brownian motion $\widetilde W$. Set
\begin{align*}
A(s;\kappa):=a(s,0)-B(s;\kappa).
\end{align*}
Since $d/ds a(s,0)=-1/(2\sigma^2)$, the ODEs in the theorem are
equivalent to
\begin{align*}
\frac{d}{ds} A(s;\kappa)&=2\lambda_s^\kappa A(s;\kappa)-2q_s^{AI}v_s^{AI}\lambda_s^\kappa A(s;\kappa)^2-\frac{1}{2\sigma^2}, \qquad A(T;\kappa)=0,\\
\frac{d}{ds} C(s;\kappa)&=-q_s^{AI}v_s^{AI}\lambda_s^\kappa A(s;\kappa), \qquad C(T;\kappa)=0. \labeleq{AC_system}
\end{align*}
Moreover,
\begin{align*}
q_s^{AI}v_s^{AI}a(s,0)=\frac{q_s^{AI}v_s^{AI}(T-s)}{2\sigma^2}=\frac{1-q_s^{AI}}{2}<\frac{1}{2}.
\end{align*}
Backward scalar comparison in the first equation of \eqref{AC_system}, using $0$ and $a(s,0)$ as comparison functions, gives
\begin{align*}
0<A(s;\kappa)<a(s,0), \qquad s\in[t,T).
\end{align*}
Consequently,
\begin{align*}
0<B(s;\kappa)<a(s,0), \qquad C(s;\kappa)=\int_s^Tq_u^{AI}v_u^{AI}\lambda_u^\kappa A(u;\kappa)du>0. \labeleq{BC_positive}
\end{align*}
Using the preceding expression for $d q_s^{AI}/ds$ and $A(s;\kappa)<a(s,0)$,
\begin{align*}
q_s^{AI}v_s^{AI}\lambda_s^\kappa A(s;\kappa)<\frac{1-q_s^{AI}}{2}\lambda_s^\kappa<\frac{1}{2}\frac{d}{ds}q_s^{AI}.
\end{align*}
Therefore,
\begin{align*}
0<C(t;\kappa)<\frac{1}{2}\int_t^T\frac{d}{ds}q_s^{AI}ds=\frac{1-q_t^{NI}}{2}<\frac{1}{2}.
\end{align*}
Now define
\begin{align*}
F(s,d)=\exp\left(A(s;\kappa)d^2+C(s;\kappa)\right).
\end{align*}
By \eqref{AC_system}, $F$ solves
\begin{align*}
F_s-\lambda_s^\kappa dF_d+\frac{1}{2}q_s^{AI}v_s^{AI}\lambda_s^\kappa F_{dd}+\frac{d^2}{2\sigma^2}F=0, \qquad F(T,d)=1.
\end{align*}
Applying the Feynman--Kac formula to
\eqref{tilted_gap_dynamics} yields
\begin{align*}
\widetilde{\E}_{t,x,m,v}^{\mu}\left[\exp\left(\frac{1}{2\sigma^2}\int_t^T\left(\Delta_s^{AI}\right)^2ds\right)\right]=\exp\left(\left(a(t,0)-B(t;\kappa)\right)\left(\Delta^{NI}\right)^2+C(t;\kappa)\right). \labeleq{gap_moment_decomposition}
\end{align*}
The required exponential integrability follows from the argument in the
proof of Lemma \ref{lemma_a3_aux}. Combining
\eqref{gap_moment_decomposition} with \eqref{j_ai_presentation} and
\begin{align*}
J^{NI}(t,x,m,v;\mu)
=V^{FI}(t,x;\mu)\exp\lb a(t,0)\lb\Delta^{NI}\rb^2\rb,
\end{align*}
and applying inverse utility, gives
\begin{align*}
P^{AI}_{SC}(\mu,\kappa)=\frac{e^{-r(T-t)}}{\gamma}\left[B(t;\kappa)\left(\Delta^{NI}\right)^2-C(t;\kappa)\right].
\end{align*}
Together with \eqref{BC_positive}, this proves the result.
\end{proof}

\subsection{Proof of Proposition \ref{propo_43}:}\label{proof_propo_43}

\begin{proof}
By \eqref{prior_belief} and \eqref{trading_strategy_gaps},
\begin{align*}
\Delta^{NI}=q_t^{NI}(m-r)-(\mu-r)\sim\N\lb-\lb1-q_t^{NI}\rb(m-r),v\rb.
\end{align*}
Hence, with $Z_t$ as in the proposition,
\begin{align*}
\lb\Delta^{NI}\rb^2\overset{d}{=}Z_t^2.
\end{align*}
Substitution into \eqref{correction_fluctuation_decomposition} yields
\begin{align*}
P_{SC}^{AI}(\mu,\kappa)\overset{d}{=}\frac{e^{-r(T-t)}}{\gamma}\left[B(t;\kappa)Z_t^2-C(t;\kappa)\right],
\end{align*}
as was to be shown.
\end{proof}

\subsection{Proof of Proposition \ref{propo_46}:} \label{proof_propo_46}

\begin{proof}
For each $\tau>0$, take terminal date $t+\tau$ and define
\begin{align*}
K_\kappa(\tau)
:=2\sigma^2\ln\widetilde{\E}^{\mu}_{t,x,m,v}
\lsb\exp\lb\frac{1}{2\sigma^2}\int_t^{t+\tau}
\lb\Delta_s^{AI}\rb^2ds\rb\rsb. \labeleq{short_horizon_K}
\end{align*}
Here the measure and processes on the right correspond to $\kappa$. By
\eqref{short_horizon_K}, \eqref{scenario_conditional_certainty_equivalents}, and
\eqref{j_ai_presentation},
\begin{align*}
CE^{AI}_{SC}(\mu;\tau)
=x+\frac{e^{-r\tau}}{2\gamma\sigma^2}
\lsb\tau(\mu-r)^2-K_\kappa(\tau)\rsb. \labeleq{short_horizon_CE}
\end{align*}
Since $\kappa=0$ gives the price observer, \eqref{short_horizon_CE} yields
\begin{align*}
P^{AI}_{SC}(\mu,\kappa;\tau)-P^{PO}_{SC}(\mu;\tau)
=-\frac{e^{-r\tau}}{2\gamma\sigma^2}
\lb K_\kappa(\tau)-K_0(\tau)\rb. \labeleq{ce_short_horizon}
\end{align*}
After setting $s=t+\tau y$, the variance ODE in \eqref{filter_ai_compact} and
the affine Gaussian dynamics \eqref{tilted_filter_dynamics} imply, uniformly
for $y\in[0,1]$,
\begin{align*}
\frac{q_{t+\tau y}^{AI}-1}{\tau}
&\longrightarrow-\frac{(1-y)v}{\sigma^2},\\
\frac{\widetilde{\E}^{\mu}_{t,x,m,v}
\lsb m_{t+\tau y}^{AI}-m\rsb}{\tau}
&\longrightarrow-y\lsb\frac{G_\kappa(v)}{v}(m-\mu)
+\frac{v}{\sigma^2}(m-r)\rsb,\\
\frac{\widetilde{\E}^{\mu}_{t,x,m,v}
\lsb\lb m_{t+\tau y}^{AI}-m\rb^2\rsb}{\tau}
&\longrightarrow yG_\kappa(v). \labeleq{short_horizon_filter_limits}
\end{align*}
Expanding the square in \eqref{delta_ai} and applying dominated
convergence to \eqref{short_horizon_filter_limits} give
\begin{align*}
&\lim_{\tau\downarrow0}\frac{1}{\tau^2}\Bigg\{
\widetilde{\E}^{\mu}_{t,x,m,v}
\lsb\int_t^{t+\tau}\lb\Delta_s^{AI}\rb^2ds\rsb
-\tau(m-\mu)^2\Bigg\}\\
&=-\frac{2v}{\sigma^2}(m-\mu)(m-r)
-\frac{G_\kappa(v)}{v}(m-\mu)^2+\frac{G_\kappa(v)}{2}. \labeleq{integrated_gap_limit}
\end{align*}
Gaussian fourth-moment bounds and local exponential integrability also give
\begin{align*}
\lim_{\tau\downarrow0}\frac{1}{\tau^2}
\Var_{\widetilde{\PP}^{\mu}}
\lb\int_t^{t+\tau}\lb\Delta_s^{AI}\rb^2ds\rb&=0,\\
\lim_{\tau\downarrow0}\frac{1}{\tau^2}\Bigg\{
K_\kappa(\tau)-\widetilde{\E}^{\mu}_{t,x,m,v}
\lsb\int_t^{t+\tau}\lb\Delta_s^{AI}\rb^2ds\rsb\Bigg\}&=0. \labeleq{short_horizon_cumulant_control}
\end{align*}
Subtracting \eqref{integrated_gap_limit} at $\kappa=0$ and using
\eqref{variance_decay_rate} and \eqref{short_horizon_cumulant_control} yield
\begin{align*}
\lim_{\tau\downarrow0}\frac{K_\kappa(\tau)-K_0(\tau)}{\tau^2}
=\frac{\kappa v}{2f(v)}\lb v-2(m-\mu)^2\rb. \labeleq{short_horizon_K_difference}
\end{align*}
Together, \eqref{ce_short_horizon} and \eqref{short_horizon_K_difference} give
\begin{align*}
\lim_{\tau\downarrow0}\frac{P^{AI}_{SC}(\mu,\kappa;\tau)
-P^{PO}_{SC}(\mu;\tau)}{\tau^2}
=-\frac{\kappa v}{4\gamma\sigma^2f(v)}
\lb v-2(m-\mu)^2\rb,
\end{align*}
which proves \eqref{short_horizon_ai}.

For the price observer, set
$q:=\frac{\sigma^2}{\sigma^2+\tau v}$ and write $C(q):=C(t;0)$. Then \eqref{correction_fluctuation_decomposition} and \eqref{b_c_po} give
\begin{align*}
P^{PO}_{SC}(\mu;\tau)
=\frac{e^{-r\tau}}{\gamma}
\lsb\frac{(1-q)a(t,0)}{2-q}\lb\Delta^{NI}\rb^2-C(q)\rsb. \labeleq{price_observer_short_horizon_exact}
\end{align*}
The required elementary limits are
\begin{align*}
\lim_{\tau\downarrow0}\frac{1-q}{\tau}
&=\frac{v}{\sigma^2},&
\lim_{\tau\downarrow0}\Delta^{NI}&=m-\mu,&
\lim_{q\uparrow1}\frac{C(q)}{(1-q)^2}&=\frac{1}{4}. \labeleq{price_observer_short_horizon_limits}
\end{align*}
Combining \eqref{price_observer_short_horizon_exact}--\eqref{price_observer_short_horizon_limits}
with $a(t,0)=\frac{\tau}{2\sigma^2}$ gives
\begin{align*}
\lim_{\tau\downarrow0}\frac{P^{PO}_{SC}(\mu;\tau)}{\tau^2}
=-\frac{v}{4\gamma\sigma^4}\lb v-2(m-\mu)^2\rb,
\end{align*}
which proves \eqref{short_horizon_po}. If
$(m-\mu)^2<\frac{v}{2}$, both limits are strictly negative, so the asserted
ordering holds for all sufficiently small $\tau>0$.
\end{proof}

\subsection{{Proof of Proposition \ref{propo_high_capacity_rates}:}}\label{proof_propo_high_capacity_rates}

\begin{proof}
For positive quantities depending on $\kappa$, write $x_\kappa\asymp y_\kappa$ when $x_\kappa/y_\kappa$ is bounded above and bounded away from zero for all sufficiently large $\kappa$. Comparisons involving $u$ are uniform over the stated range of $u$. All comparison constants may depend on the fixed state, benchmark, and model parameters but not on $\kappa$ or $\mu$.

Set $\tau:=T-t$, $u:=s-t\in[0,\tau]$, and suppress fixed initial arguments. The Shannon and Tsallis specifications can be written jointly as
\begin{align*}
f^\rho(y)=c_\rho y^{1-\beta_\rho},\qquad
\beta_\rho:=\frac{\rho-1}{2},\qquad
c_\rho:=\frac{1}{2}\rho^{-1/2}(2\pi)^{(1-\rho)/2},
\end{align*}
where $\rho=1$ gives the Shannon specification. Write $v_\kappa(u):=v_{t+u}^{AI}$ and $q_\kappa(u):=q_{t+u}^{AI}$. The variance equation becomes
\begin{align*}
\frac{d}{du}v_\kappa(u)
=-\frac{v_\kappa(u)^2}{\sigma^2}
-\frac{\kappa}{c_\rho}v_\kappa(u)^{1+\beta_\rho},
\qquad v_\kappa(0)=v.
\labeleq{power_variance_ode}
\end{align*}
If $\beta_\rho=0$, scalar comparison gives
\begin{align*}
v_\kappa(u)\asymp\exp\lb-\frac{\kappa u}{c_\rho}\rb,
\qquad u\in[0,\tau].
\end{align*}
If $\beta_\rho>0$, differentiating $v_\kappa^{-\beta_\rho}$ gives
\begin{align*}
\frac{d}{du}v_\kappa(u)^{-\beta_\rho}
=\frac{\beta_\rho\kappa}{c_\rho}
+\frac{\beta_\rho}{\sigma^2}v_\kappa(u)^{1-\beta_\rho}.
\end{align*}
The second term is bounded when $\beta_\rho\leq1$ and sublinear in $v_\kappa^{-\beta_\rho}$ when $\beta_\rho>1$. Scalar comparison on $[0,\tau]$ therefore yields
\begin{align*}
v_\kappa(u)\asymp(1+\kappa u)^{-\frac{1}{\beta_\rho}},
\qquad u\in[0,\tau].
\labeleq{power_variance_comparison}
\end{align*}
All comparison constants may depend on $(\tau,v,\sigma,\rho)$ but not on $\kappa$ or $u$. Hence, for every $p>0$,
\begin{align*}
\int_0^\tau v_\kappa(u)^pdu
\asymp
\begin{cases}
\kappa^{-1},&\beta_\rho=0\text{ or }0<\beta_\rho<p,\\
\frac{\ln(1+\kappa)}{\kappa},&\beta_\rho=p,\\
\kappa^{-\frac{p}{\beta_\rho}},&\beta_\rho>p.
\end{cases}
\labeleq{variance_integral_rates}
\end{align*}
As in the proof of Theorem \ref{theorem_42}, write $A(s;\kappa):=a(s,0)-B(s;\kappa)$, and set
\begin{align*}
w_\kappa(u):=q_\kappa(u)v_\kappa(u),\qquad
\lambda_\kappa(u):=\lambda_{t+u}^\kappa.
\end{align*}
Since $q_t^{NI}\leq q_\kappa(u)\leq1$, $w_\kappa(u)\asymp v_\kappa(u)$ uniformly. The first equation in \eqref{AC_system} and variation of constants give
\begin{align*}
A(t;\kappa)=\frac{1}{2\sigma^2}\int_0^\tau
\exp\lb-2\int_0^u\lambda_\kappa(y)\lsb1-w_\kappa(y)A(t+y;\kappa)\rsb dy\rb du.
\labeleq{A_variation_constants}
\end{align*}
By \eqref{correction_fluctuation_decomposition} and \eqref{p_sc_limit},
\begin{align*}
P_{SC}^{AI,\infty}(\mu)-P_{SC}^{AI}(\mu,\kappa)
=
\frac{e^{-r(T-t)}}{\gamma}
\lsb
A(t;\kappa)\lb\Delta^{NI}\rb^2+C(t;\kappa)
\rsb.
\end{align*}
By Lemma \ref{lemma_a1}, $C(t;\kappa)\to0$ as $\kappa\to\infty$. Moreover, \eqref{BC_positive} implies
\begin{align*}
0\leq\int_0^u w_\kappa(y)\lambda_\kappa(y)A(t+y;\kappa)dy\leq C(t;\kappa)\to0,
\qquad u\in[0,\tau].
\end{align*}
Hence, the nonlinear contribution to the exponent in \eqref{A_variation_constants} converges uniformly to zero. Direct differentiation of
\begin{align*}
\frac{1}{w_\kappa(u)}
=
\frac{1}{v_\kappa(u)}
+
\frac{\tau-u}{\sigma^2}
\end{align*}
and \eqref{power_variance_ode} gives
\begin{align*}
\frac{d}{du}w_\kappa(u)
&=-\frac{\kappa w_\kappa(u)^2}{f^\rho(v_\kappa(u))},\\
\lambda_\kappa(u)
&=\frac{w_\kappa(u)}{\sigma^2}
-\frac{d}{du}\ln w_\kappa(u).
\end{align*}
Consequently,
\begin{align*}
\exp\lb-2\int_0^u\lambda_\kappa(y)dy\rb
=
\lb\frac{w_\kappa(u)}{w_\kappa(0)}\rb^2
\exp\lb-\frac{2}{\sigma^2}\int_0^u w_\kappa(y)dy\rb.
\end{align*}
The last exponential is bounded above and bounded away from zero uniformly. Since $w_\kappa(0)=q_t^{NI}v$ is independent of $\kappa$, it follows from \eqref{A_variation_constants} that
\begin{align*}
A(t;\kappa)\asymp\int_0^\tau v_\kappa(u)^2du.
\labeleq{A_variance_integral}
\end{align*}

The identity established in the proof of Lemma \ref{lemma_a1} gives
\begin{align*}
\frac{1}{2\sigma^2}\int_0^\tau w_\kappa(u)du
=
C(t;\kappa)
-\frac{1}{2}\ln\lb1-2w_\kappa(0)A(t;\kappa)\rb.
\labeleq{C_variance_identity}
\end{align*}
The left-hand side has the order of $\int_0^\tau v_\kappa(u)du$, while the logarithmic term has the order of $A(t;\kappa)$. If $\beta_\rho\geq1$, \eqref{variance_integral_rates} and \eqref{A_variance_integral} show that $A(t;\kappa)$ has strictly smaller order than $\int_0^\tau v_\kappa(u)du$. Therefore,
\begin{align*}
C(t;\kappa)\asymp\int_0^\tau v_\kappa(u)du.
\end{align*}

If $0\leq\beta_\rho<1$, restrict the integral representation in \eqref{BC_positive} to
\begin{align*}
u\in\lsb\frac{c_1}{\kappa},\frac{c_2}{\kappa}\rsb
\end{align*}
for fixed $0<c_1<c_2$. On this interval, scalar comparison gives $w_\kappa(u)$ bounded away from zero, $\lambda_\kappa(u)\asymp\kappa$, and $A(t+u;\kappa)\asymp\kappa^{-1}$. This gives the matching lower bound, while \eqref{C_variance_identity} gives the upper bound. Hence
\begin{align*}
C(t;\kappa)\asymp\int_0^\tau v_\kappa(u)du
\labeleq{C_variance_integral}
\end{align*}
for every $\beta_\rho\geq0$.

Applying \eqref{variance_integral_rates} with $p=2$ and $p=1$, and substituting $\beta_\rho=(\rho-1)/2$, gives
\begin{align*}
A(t;\kappa)&\asymp
\begin{cases}
\kappa^{-1},&1\leq\rho<5,\\
\frac{\ln(1+\kappa)}{\kappa},&\rho=5,\\
\kappa^{-4/(\rho-1)},&\rho>5,
\end{cases}\\
C(t;\kappa)&\asymp r_\rho(\kappa).
\labeleq{coefficient_high_capacity_rates}
\end{align*}
In particular, $A(t;\kappa)$ is bounded above by a positive multiple of $r_\rho(\kappa)$ for all sufficiently large $\kappa$.

Finally, as $\mu$ ranges over $\mathbb{R}$, $z:=\lb\Delta^{NI}\rb^2$ ranges over $[0,\infty)$. By \eqref{correction_fluctuation_decomposition} and \eqref{p_sc_limit},
\begin{align*}
\sup_{\mu\in\mathbb{R}}
\frac{
P_{SC}^{AI,\infty}(\mu)-P_{SC}^{AI}(\mu,\kappa)
}{
P_{SC}^{AI,\infty}(\mu)-P_{SC}^{AI}(\mu,\bar\kappa)
}
&=
\sup_{z\geq0}
\frac{A(t;\kappa)z+C(t;\kappa)}
{A(t;\bar\kappa)z+C(t;\bar\kappa)}
=
\max\left\{
\frac{A(t;\kappa)}{A(t;\bar\kappa)},
\frac{C(t;\kappa)}{C(t;\bar\kappa)}
\right\}.
\end{align*}
Because $t<T$ and $\bar\kappa$ is finite, Theorem \ref{theorem_42} gives
\begin{align*}
A(t;\bar\kappa)>0,\qquad C(t;\bar\kappa)>0.
\end{align*}
Equation \eqref{coefficient_high_capacity_rates} implies that $C(t;\kappa)$ is bounded above and below by positive multiples of $r_\rho(\kappa)$, while $A(t;\kappa)$ is bounded above by a positive multiple of $r_\rho(\kappa)$. Dividing by the fixed positive benchmark coefficients and taking the maximum yields constants $0<K_1\leq K_2<\infty$ and $\kappa_0>\bar\kappa$ for which \eqref{voi_high_capacity_rates} holds.

Since $r_\rho(\kappa)\to0$, the upper bound in \eqref{voi_high_capacity_rates} is smaller than one for all sufficiently large $\kappa$. Hence,
\begin{align*}
P_{SC}^{AI,\infty}(\mu)-P_{SC}^{AI}(\mu,\kappa)
<
P_{SC}^{AI,\infty}(\mu)-P_{SC}^{AI}(\mu,\bar\kappa)
\end{align*}
uniformly in $\mu$, which proves \eqref{eventual_capacity_dominance}.
\end{proof}

\end{document}